\newcommand{\CLASSINPUTtoptextmargin}{1.75cm}
\newcommand{\CLASSINPUTbottomtextmargin}{1.75cm}
\newcommand{\CLASSINPUToutersidemargin}{1.7cm}

\documentclass[10pt,twocolumn,final,journal]{IEEEtran}

\usepackage{amsmath, mathrsfs, mathtools}
\usepackage{amsfonts}
\usepackage{amssymb}
\usepackage{amsthm}
\usepackage{dsfont}
\usepackage{graphicx}
\usepackage{color}
\usepackage{graphicx}
\usepackage{subfig}
\usepackage{caption}
\usepackage[backend=bibtex,giveninits=true,sorting=none]{biblatex}
\usepackage{balance}
\renewcommand*{\bibfont}{\footnotesize}

\let\labelindent\relax
\usepackage{enumitem}

\newcommand{\subparagraph}{}
\usepackage[compact]{titlesec}
\titlespacing*{\section}{15pt}{1.2\baselineskip}{0.9\baselineskip}

\theoremstyle{remark}
\newtheorem{theorem}{\bf Theorem}
\newtheorem{assumption}{\bf Claim}
\newtheorem{corollary}{\bf Corollary}
\newtheorem{remark}{\bf Remark}
\newtheorem{lemma}{\bf Lemma}

\long\def\comment#1{}

\newcommand\figref{Figure~\ref}

\newcommand{\ben}{\begin{enumerate}}
\newcommand{\een}{\end{enumerate}}

\newcommand{\beq}{\begin{equation}}
\newcommand{\eeq}{\end{equation}}

\newcommand{\bi}{\begin{itemize}}
\newcommand{\ei}{\end{itemize}}

\DeclareMathOperator*{\argmin}{arg\,min}
\DeclareMathOperator*{\argmax}{arg\,max}

\newcommand{\PP}{\mathbb{P}}
\newcommand{\RR}{\mathbb{R}}

\newcommand{\EE}{\mathbb{E}}

\newcommand{\av}{{\bf a}}

\newcommand{\cv}{{\bf c}}

\newcommand{\ev}{{\bf e}}

\newcommand{\mv}{{\bf m}}

\newcommand{\rv}{{\bf r}}
\newcommand{\sv}{{\bf s}}

\newcommand{\xv}{{\bf x}}
\newcommand{\yv}{{\bf y}}
\newcommand{\zv}{{\bf z}}

\newcommand{\Am}{{\bf A}}

\newcommand{\Id}{{\bf I}}

\newcommand{\Zm}{{\bf Z}}

\newcommand{\Lc}{{\cal L}}

\newcommand{\Sc}{{\cal S}}

\newcommand{\alphav}{\hbox{\boldmath$\alpha$}}

\newcommand{\thetav}{\hbox{\boldmath$\theta$}}

\newcommand{\piv}{\hbox{\boldmath$\pi$}}
\newcommand{\rhov}{\hbox{\boldmath$\rho$}}

\newcommand{\SNR}{{\sf SNR}}

\title{SPARCs for Unsourced Random Access}
\author{Alexander Fengler, Peter Jung, Giuseppe Caire
\thanks{The authors are with the Communications and Information Theory Group,
Technische Universit\"{a}t Berlin (\{fengler, peter.jung,
caire\}@tu-berlin.de).}
\thanks{Parts of this paper were presented in
the 2019 and 2020 IEEE International Symposium on Information Theory (ISIT) \cite{Fen2019c}.}
}

\begin{document}

\maketitle

\begin{abstract}
    Unsourced random-access (U-RA) is a type of grant-free random access with a
    virtually unlimited number of users, of which only a certain number $K_a$ are
    active on the same time slot. Users employ exactly the same codebook, and
    the task of the receiver is to decode the list of transmitted messages.
    We present a concatenated coding
    construction for U-RA on the AWGN channel,
    in which a sparse regression code (SPARC) is used as an
    inner code to create an effective outer OR-channel. Then an outer code is used
    to resolve the multiple-access interference in the OR-MAC. We propose a modified
    version of the approximate message passing (AMP) algorithm as an inner decoder
    and give a precise asymptotic analysis of the error probabilities of the AMP decoder
    and of a hypothetical optimal
    inner MAP decoder. This analysis shows that the
    concatenated construction under optimal decoding can achieve a vanishing per-user error probability
    in the limit of large blocklength and a large number of active users at sum-rates up to the symmetric
    Shannon capacity, i.e. as long as $K_aR < 0.5\log_2(1+K_a\SNR)$.
    This extends previous point-to-point optimality results about SPARCs to the unsourced
    multiuser scenario.
    Furthermore, we give an optimization algorithm to find the power allocation for the inner
    SPARC code that minimizes the $\SNR$ required to achieve a given target per-user error probability with 
    the AMP decoder.
\end{abstract}


\section{Introduction}
\label{sec:intro}
In new application scenarios of wireless networks such as Internet-of-Things (IoT),
it is envisioned that
a very large number of devices (referred to as users) are sending data to a common
access point. 
Typical examples thereof include sensors for monitoring smart infrastructure or biomedical devices. 
This type of communication is characterized by short messages and sporadic activity.
The large number of users and the sporadic nature of the transmission makes
it very wasteful to allocate dedicated transmission resources to all the users.
In contrast to these requirements,
the traditional information theoretic treatment of the multiple-access uplink channel
is focused on few users $K$, large blocklength $n$ and coordinated transmission,
in the sense that each user is given an individual distinct
codebook, and the $K$ users agree on which rate $K$-tuple inside the capacity
region to operate \cite{Ahl1973,Lia1973,El1980}.
Mathematically, this is reflected by considering the limit of infinite message- and
blocklength while keeping the rate and the number of users fixed.
This approach does not capture the bursty random arrival of messages in real world
multiple access networks \cite{Gal1985}, which lead to the widespread success
of packet based random-access models \cite{Abr1970,Cas2007,Liv2011a,Pao2015}. 
Such models are based on simplified collision channels \cite{Mas1985},
which ignore the underlying physical communication and are thereby
limited in the achievable performance \cite{Eph1998,Ber2016}.
An alternative information theoretic route,
more suited to capture the short messages of machine-type communication,
was taken in recent works like \cite{Che2017,Pol2017}, where the number of users $K$
is taken to infinity along with the blocklength. 
It was shown that the information theoretic limits
may be drastically different
when the number of users grows together with the blocklength.

On a high level, we distinguish between {\em grant-based} and
{\em grant-free} approaches.
In a grant-based protocol the active users are identified and the base-station (BS)
can then allocate transmission resources to the active users, while in a grant-free protocol the
users transmit their data right away without awaiting the approval of the BS.
For a recent overview of grant-free protocols see \cite{Che2021a}.
A novel grant-free random access paradigm, referred to as unsourced random-access (U-RA),
was suggested in \cite{Pol2017}. In U-RA each user employs the same codebook and the task
of the decoder is to recover the list of transmitted messages irrespective of the identity of
the users. The number of \emph{inactive}
users in such a model can be arbitrarily large and the performance of the system depends
only on the number of \emph{active} users $K_a$. Furthermore, 
a transmission protocol without the need for
a subscriber identity
is well suited for mass production. These features make U-RA particularly interesting for
the aforementioned IoT applications. 

In \cite{Pol2017} the U-RA model for the real adder AWGN-MAC was introduced and
a finite-blocklength random coding bound on the achievable energy-per-bit over noise power spectral density
($E_b/N_0$)
was established.
In following works several
practical approaches were suggested which successively reduced the
gap to the random coding achievability bound \cite{Vem2017,Ama2020a,Mar2019,Pra2020a}.
The model has been extended to fading \cite{Kow2019a} and MIMO channels \cite{Fen2021c}.
A concatenated coding approach for the U-RA problem on the real adder AWGN was proposed
in \cite{Ama2020a}. The idea is to split each transmission 
into $L$ subslots. In each subslot the active users send a column from a common
\emph{inner} coding matrix, while the symbols across all subslots are chosen from a
common \emph{outer tree code}. 
We build upon the finding of \cite{Ama2020a} and its similarity to sparse regression codes
(SPARCs) to give an improved inner decoder and a complete asymptotic error analysis.
SPARCs were introduced in \cite{Jos2012} as a class of channel codes for the point-to-point
AWGN channel, which can achieve rates up to Shannon capacity under maximum-likelihood decoding.
Later, it was shown that SPARCs can achieve capacity under
approximate message passing (AMP) decoding with either power
allocation \cite{Rus2017} or spatial coupling \cite{Bar2016a, Bar2017a, Hsi2018a, Rus2020}.
AMP is an iterative low-complexity algorithm for solving random linear
estimation problems or generalized versions thereof \cite{Don2009a,Ran2011,Bay2011}.
A recent survey on SPARCs can be found in \cite{Ven2019a}.

One of the appealing features of the AMP algorithm is that it is possible to analyse its
asymptotic error probability, averaged over certain random matrix ensembles,
through the so called state evolution (SE) equations \cite{Bay2011,Ber2020}.
Interestingly the SE equations can also be obtained as the stationary points of the
replica symmetric (RS) potential, an expression that was first calculated through the non-rigorous
replica method \cite{Tan2002,Guo2005c}. It was shown that
in random linear estimation problems the stationary points of the RS potential also
characterize the symbols-wise posterior distribution of the input elements and therefore
also the error probability of several optimal estimators like the minimum-mean-square error (MMSE)
estimator \cite{Guo2005c,Guo2009}.
The difference between the AMP and the MMSE estimate is that
the MMSE estimate always corresponds to the global minimum of the RS-potential, while the
AMP algorithm gets `stuck' in local minima. The rate below which a local minimum appears
was called the \emph{algorithmic} or \emph{belief-propagation} threshold in \cite{Guo2009, Krz2012a,Bar2017a}.
It was shown in \cite{Krz2012a,Bar2020b} that, despite the existence of local minima in the RS-potential,
the AMP algorithm can
still converge to the global minimum when used with spatially coupled matrices.
Although the RS-potential was derived by (and named after) the non-rigorous replica method,
it was recently proven to hold rigorously \cite{Ree2016,Bar2019}. The proof
of \cite{Bar2019} is more general in the sense that it includes the case
where the unknown, to be estimated, vector of information symbols
consists of blocks of size $2^J$ and each block is considered to be drawn iid from some distribution
on $\RR^{2^J}$.\\
Our main contribution in this work are as follows
\bi
\item We extend the concept of sparse regression codes to the unsourced random access setting
    by making use of the tree code of \cite{Ama2020a}. 
\item For the resulting inner-outer concatenated coding scheme,
    we introduce a matching outer channel model,
    analyse the achievable rates on this outer channel, and compare it to existing practical solutions.
\item We propose a modified approximate message passing algorithm as an inner decoder and
    analyse its asymptotic error probability through its SE.
    We use the connections between SE and the RS formula to find
    the error probability of a hypothetical MAP decoder. 
\item We find that the error probability of the inner decoder admits a simple
    closed form in the limit of $K_a,J\to\infty$ with $J = \alpha \log_2 K_a$
    for some $\alpha >1$. The limit was also considered in \cite{Ama2020a}, motivated by the
    fact that $J$ is the number of bits required to encode the identity of each of up to $K_\text{tot}$
    users if $K_a = K_\text{tot}^{1/\alpha}$.
    We show that the per-user error probability of the concatenated scheme with inner MAP decoding
    vanishes in the
    limit of large blocklength and infinitely many users,
    if the sum-rate is smaller than the symmetric Shannon capacity
    $0.5\log_2(1+K_a\SNR)$. 
    This shows that an unsourced random access scheme can,
    even with no coordination between users,
    achieve the same symmetric
    rates as a non-unsourced scheme. 
\item Using the results from the asymptotic analysis we identify parameter regions
    where the AMP decoder can achieve the same error probability as the MAP decoder. In parameter regions
    where there is a gap between the achievable error probability of the AMP and the MAP decoder we propose a method
    for finding an optimal power allocation that is able to improve the performance of the AMP decoder significantly.
\item We provide finite-length simulations to show the efficiency of the proposed coding scheme and the accuracy
    of the analytical predictions.
\ei
The paper is organized as follows.
In Section \ref{sec:model} we describe the channel model.
In Section \ref{sec:coding} we introduce the concatenated coding scheme.
In Section \ref{sec:inner} we introduce the inner AMP decoder and the optimal, but uncomputable,
MAP decoder and analyse their asymptotic error probabilities.
In Section \ref{sec:decision} we analyse the quantization step, which is necessary for a binary-input
outer decoder.
In Section \ref{sec:outer} we formulate the outer channel and give converse and achievability results.
In Section \ref{sec:conc} we analyse the concatenated code.
In Section \ref{sec:pa} we give an algorithm to optimize the power allocation and
in Section \ref{sec:practical} we introduce a low-complexity approximation of the suggested 
AMP algorithm.
In Section \ref{sec:sims} we give finite-length simulations and compare them to the analytical results.

\section{Channel Model}
\label{sec:model}
Let $K_a$ denote the number of active users,
$n$ the number of available channel uses and
$B = nR$ the size of a message in bits.
The spectral efficiency is given by $\mu = K_aR$.
The channel model used is
\beq
\yv = \sum_{i=1}^{K_\text{tot}} q_i \xv_i + \zv, 
\label{eq:basic_channel}
\eeq
where each $\xv_i \in \mathcal{C} \subset \RR^n$
is taken from a common codebook $\mathcal{C}$ and $q_i\in\{0,1\}$ are binary variables
indicating whether a user is active.
The number of active users is denoted as $K_a = \sum_{i=1}^{K_\text{tot}}q_i$.  
The codewords are assumed to be normalized as
$\|\xv_i\|_2^2 = nP$ for a given energy-per-symbol $P$, and the noise vector $\zv$ is Gaussian iid
$z_i\sim\mathcal{N}(0,N_0/2)$,
such that $\SNR = 2P/N_0$ denotes the per-user $\SNR$.
All the active users pick one of the $2^B$ codewords from $\mathcal{C}$,
based on their message $W_k\in[1:2^B]$.
The decoder of the system produces a list $g(\yv)$ of at most $K_a$ messages.
An error is declared if one of the transmitted messages is missing in the output list $g(\yv)$
and we define the per-user probability of error as:
\beq
P_e = \frac{1}{K_a}\sum_{k=1}^{K_a} \PP(W_k \notin g(\yv)).
\eeq
Note that the error
is independent of the user identities in general and especially independent
of the inactive users.
The performance of the system is measured in terms of the required
$E_b/N_0 := P/(RN_0)$ for a target $P_e$ and the described coding construction is called \emph{reliable}
if $P_e \to 0$ as $n\to\infty$.
\section{Concatenated Coding}
\label{sec:coding}
In this work we focus on a special type of codebook,
where each transmitted codeword is created in the following way:
First, the $B$-bit message $W_k$ of user $k$ is mapped to an $LJ$-bit codeword from some common \emph{outer}
codebook. Then each of the $J$-bit sub-sequences is mapped to an index $i_k(l) \in [1:2^J]$
for $l=[1:L]$ and $k=[1:K_a]$. The inner codebook is based on a set of $L$
coding matrices $\Am_l\in\RR^{n\times 2^J}$.
Let $\av^{(l)}_i$ with
$i=[1:2^J]$ denote the columns of $\Am_l$. The inner codeword of user $k$
corresponding to the sequence of indices $i_k(1),...,i_k(L)$ is then created as
\beq
\xv_k = \sum_{l=1}^L \sqrt{P_l}\av^{(l)}_{i_k(l)}.
\label{eq:inner_encoding}
\eeq
The columns $\av^{(l)}_i$ of $\Am_l$ are assumed to be zero mean and
scaled such that $\|\av^{(l)}_i\|_2^2 = 1$
and the power coefficients $P_l$ are chosen such that $\EE[\|\xv_k\|_2^2] \leq nP$ where
the expectation is taken over all choices of indices $(i_k(1),...,i_k(L))$.
The above encoding model can be written in matrix form as
\beq
\yv = \sum_{k=1}^{K_a} \Am \mv_k + \zv = \Am \left(\sum_{k=1}^{K_a} \mv_k\right) + \zv.
\label{eq:inner_channel}
\eeq
where $\Am = (\Am_1|...|\Am_L)$ and $\mv_k \in \RR^{L2^J}$ is a non-negative
vector satisfying
$m_{k,(l-1)2^J + i_k(l)} = \sqrt{P_l}$ and zero otherwise, for all $l=[1:L]$.
Let $\thetav = \sum_{k=1}^{K_a} \mv_k$ and let $\sv$ denote the support of $\thetav$ with multiplicity,
such that $\thetav = \left(\sqrt{P_1}\sv^1|...|\sqrt{P_L}\sv^L\right)^\top$.
That is, the components of $\sv^l$
indicate how many active users have chosen a specific column of $\Am_l$.
We refer to $\sqrt{P_l}\sv^l$ as the $l$-th \emph{section} of $\thetav$.
The linear structure allows to write the channel \eqref{eq:inner_channel} as a concatenation of the
inner point-to-point channel $\thetav\to \Am\thetav + \zv$ and the
outer binary input adder MAC $(\mv_1,...,\mv_{K_a}) \to \sv$.
We will refer to those as the
\emph{inner} and \emph{outer channel}, the corresponding encoder and decoder
will be referred to as \emph{inner} and \emph{outer encoder/decoder} and
the aggregated system of inner and outer encoder/decoder as the \emph{concatenated system}.
The per-user inner rate in terms of bits per channel use (c.u.) is given by
$R_\text{in} := LJ/n$ and the outer rate is given by $R_\text{out} = B/LJ$.
For the sake of the analysis we shall consider a random ensemble of codes where
the matrices $\Am_l$ are generated with i.i.d. Gaussian components
$\sim\mathcal{N}(0,1/n)$ and the outer encoded indices $i_k(l)$
are distributed uniformly and independently over $[1:2^J]$.
Furthermore, we assume that the
power coefficients are uniformly $P_l \equiv nP/L$, such that the power constraint 
$\EE[\|\xv_k\|_2^2] \leq nP$ is fulfilled on average over the code ensemble.
The uniform power allocation will be relaxed
in Section \ref{sec:pa}, where we will consider a non-uniform power allocation.
\section{Inner Channel}
\label{sec:inner}
In this section we focus on the inner decoding problem of recovering $\sv$ from
\beq
\yv = \Am \thetav + \zv = \sqrt{\hat{P}}\Am \sv + \zv
    \label{eq:yAs_plus_z}
\eeq
where $\hat{P} = nP/L$.
Let $k_i \in [0:K_a]$ for $i \in [1:2^J]$
be non-negative integers.
The probability of observing a specific $\sv^l$ is given by:
\beq
p\left(\sv^l = (k_1,...,k_{2^J})^\top\right) = 2^{-K_aJ}\frac{K_a!}{k_1!\cdots k_{2^J}!}
\label{eq:multinomial}
\eeq 
if $\sum_{i=1}^{2^J} k_i = K_a$ and zero otherwise.
This is a multinomial distribution with uniform event probabilities. The marginals of such a distribution are known to
be Binomial, i.e.:
\beq
p_k := \PP(s^l_i=k) = {K_a\choose k}2^{-kJ}(1-2^{-J})^{K_a-k}
\label{eq:binomial}
\eeq
and specifically, the probability of observing a zero is:
\beq
p_0 := \PP(s^l_i = 0) = (1-2^{-J})^{K_a}.
\label{eq:p0}
\eeq
We define two estimators for $\sv$. The first is a variant of the approximate message passing (AMP) algorithm,
which we will refer to as AMP-estimator.
An estimate of $\sv$ is obtained by iterating the following equations:
\beq
\begin{split}
    \thetav^{t+1} &= f_t(\Am^\top\zv^t + \thetav^t) \\
    \zv^{t+1}   &= \yv - \Am\thetav^{t+1} +
    \frac{2^JL}{n}\zv^{t}\langle f^\prime_t(\Am^\top\zv^t+\thetav^t)\rangle
\end{split}
\label{eq:amp}
\eeq
where the function $f_t:\RR^{2^JL}\to\RR^{2^JL}$ is defined componentwise
$f_t(\xv) = (f_{t,1}(x_1),...,f_{t,2^JL}(x_{2^JL}))^\top$ and each component is given by

\beq
f_{t,i}(x) = \frac{\sqrt{\hat{P}}}{Z(x)}\sum_{k=0}^{K_a}p_kk\exp\left(\frac{1}{2\tau_t^2}\left(x-k\sqrt{\hat{P}}\right)^2\right)
    \label{eq:eta_add}
\eeq
with $\tau^2_t = \|\zv^t\|_2^2/n$, $p_0$ as in \eqref{eq:p0} and
\beq
Z(x) = \sum_{k=0}^{K_a}p_k\exp\left(\frac{1}{2\tau_t^2}\left(x-k\sqrt{\hat{P}}\right)^2\right).
\eeq 
$\langle\xv\rangle = (\sum_{i=1}^N x_i)/N$ in \eqref{eq:amp} denotes the arithmetic mean of a vector, $f^\prime_t$
denotes the componentwise derivative of $f_t$ and
$\zv^0 = \mathbf{0}$ is chosen as the initial value.
After the equations \eqref{eq:amp}
are iterated for some fixed number of iterations $T_\text{max}$,
a final estimate of $\sv$ is obtained by quantizing $\thetav^{T_\text{max}}$ to the nearest
integer multiple of $\sqrt{\hat{P}}$ and dividing by $\sqrt{\hat{P}}$.
Note, that each of the functions $f_{t,i}$ in \eqref{eq:eta_add} is chosen as the posterior-mean estimator (PME)
of the component $\theta_i$ in a scalar Gaussian channel with noise variance $\tau_t^2$. 
This is justified by the remarkable property of the AMP algorithm
that the terms $\Am^\top\zv^t + \thetav^t$ are distributed approximately
like $\mathcal{N}(\thetav, \tau_t^2\Id)$, i.e. like the true signal in iid Gaussian noise \cite{Bay2011,Ber2020}.

The second estimator that we analyse is the symbol-by-symbol maximum-a-posteriori (SBS-MAP)
estimator of $\sv$
\beq
\hat{s}^l_i = \argmax_{s \in [0:K_a]} \PP(s^l_i = s|\yv,\Am),
\label{eq:sbs-map}
\eeq
which minimizes the SBS error probability $\PP(\hat{s}^l_i \neq s^l_i)$ but is infeasible
to compute in practice.
Let 
\beq
P_e^\text{in} = \frac{1}{L2^J}\sum_{l=1}^{L}\sum_{j=1}^{2^J} \PP(\hat{s}^l_j \neq s^l_j)
\eeq
denote the inner per-user SBS error rate for some symbol-wise estimator $\hat{s}^l_j$,
and let $P_e^\text{MAP}$ and $P_e^\text{AMP}$
denote the corresponding inner error rates
of the MAP and the AMP estimator respectively.
\subsection{Asymptotic Error Analysis}
\label{subsec:asymp}
\newcommand{\Ein}{\mathcal{E}_\text{in}}
\newcommand{\Eopt}{\mathcal{E}^\text{opt}_\text{in}}
\newcommand{\Ealg}{\mathcal{E}^\text{alg}_\text{in}}
The error analysis is based on the
self-averaging property of the random linear recovery problem
\eqref{eq:yAs_plus_z} in the asymptotic limit $L,n\to\infty$ with
a fixed $J$ and fixed $R_\text{in}$. That is, although $\Am, \sv$ and $\zv$ are random
variables, the error probability of both mentioned estimators converges sharply
to its average value. The convergence behavior is fully characterized by the external parameters
$J, R_\text{in}, \SNR$ and $K_a$.
It is known that the asymptotic estimation error of the AMP algorithm can be
analysed by the so called \emph{state evolution} (SE) equations \cite{Bay2011}. 
\begin{theorem}
    \label{thm:se_amp}
    In the limit $n,L \to \infty$ for fixed $J$ and fixed $R_\text{in}$ the mean-square-error (MSE)
    of the AMP estimate \eqref{eq:amp} converges to the MSE
    of estimating $s$ in the scalar Gaussian channel
    \beq
        r = \sqrt{\eta \hat{P}} s + z
        \label{eq:decoupled_scalar}
    \eeq 
    where $s$ is distributed according to the binomial distribution $p(s=k)$
    specified in \eqref{eq:binomial}, i.e. the \emph{marginal}
    distribution of a single section of $\sv$,
    and $z\sim\mathcal{N}(0,1)$ is independent
    of $s$. 
    The factor $\eta\geq0$ is given as the smallest non-negative
    solution of 
    \beq
        \frac{\mathrm{d}}{\mathrm{d}\eta} i^\text{RS}_\text{AMP}(\eta) = 0
    \eeq
    where $i^\text{RS}_\text{AMP}(\eta)$ is given by
    \beq
    i^\text{RS}_\text{AMP}(\eta) = 2^JI(\eta\hat{P}) + \frac{2^J}{2\beta}[(\eta-1)\log_2(e)-\log_2(\eta)]
        \label{eq:RS-scalar-potential}
    \eeq 
    $I(\eta\hat{P})$ denotes the mutual information between $r$ and $s$
    in the scalar Gaussian channel \eqref{eq:decoupled_scalar}. 
    $\hfill\square$
\end{theorem}
\begin{proof}
    The theorem is merely a restatement of the SE result in \cite{Bay2011} which states that the MSE
    of $\thetav^t$ can be described asymptotically by a scalar Gaussian channel where the effective noise
    variance follows the recursion
    \beq
    \tau^2_{t+1} = 1 + \beta\hat{P}\EE[(s - f_t(\hat{P}s+\tau_t^2 z))^2]
    \eeq 
    with $\tau_0^2 = \|\yv\|_2^2/n$ and $s,z$ jointly distributed as in \eqref{eq:decoupled_scalar}.
    It is well known that in general the minimum-mean-square error (MMSE) can be
    achieved by the PME $\EE[s|y]$.
    Since $f_t$
    was chosen precisely as the PME of $\sqrt{\hat{P}}s$ in a scalar Gaussian channel of the form
    \eqref{eq:decoupled_scalar}
    with $\eta = 1/\tau_t^2$,
    it holds that
    \beq
    \EE[(s - f_t(\hat{P}s+\tau_t^2 z))^2] = \text{mmse}(\eta\hat{P}),
    \eeq
    where we introduce the MMSE function
    $\text{mmse}(\eta\hat{P}) = \EE[(s - \EE[s|y])^2]$ in a Gaussian channel \eqref{eq:decoupled_scalar}.
    Since the recursion starts at a high noise (low $\eta$) point and because
    $\text{mmse}(\eta\hat{P})$ is monotonically decreasing in $\eta$ the point of convergence is given by
    the smallest $\eta$ for which
    \beq
        \eta^{-1} = 1 + \beta\hat{P}\text{mmse}(\eta\hat{P})
        \label{eq:SE}
    \eeq 
    holds. 
    Setting the derivative of \eqref{eq:RS-scalar-potential} with respect to $\eta$ to zero
    shows that \eqref{eq:SE} is precisely the condition for a stationary point.
    This can be seen using the I-MMSE theorem \cite{Guo2005a}, which states that in a Gaussian additive
    noise channel
    it holds that:
    \beq
    \frac{1}{\log_2(e)}\frac{\mathrm{d}}{\mathrm{d}\eta}I(\eta\hat{P})
     = \frac{\hat{P}}{2}\text{mmse}(\eta\hat{P})
    \eeq
\end{proof}
The mutual information in \eqref{eq:RS-scalar-potential} is bounded by the entropy of the input distribution
for $\eta\to 0$ and so $\lim_{\eta\to 0}i^{RS}(\eta) = \infty$. Furthermore, \eqref{eq:RS-scalar-potential}
is continuously differentiable and
therefore the smallest stationary point, i.e. the smallest $\eta$ for which $\mathrm{d} i^\text{RS}(\eta)/\mathrm{d} \eta = 0$, is necessarily
either a local minimum or a saddle point but never a local maximum. 
Note that although the assumed distribution on $\sv$ is not iid the AMP algorithm in
\eqref{eq:amp} uses a separable denoiser. In this case the SE result of \cite{Bay2011} does
not require $\sv$ to be iid, but only that its empirical marginal distributions converge to some limit,
which is required
for the calculation of the SE. In fact, the presented AMP algorithm \eqref{eq:amp} is not optimal
since it does not make full use of the distribution of $\sv$, i.e. it ignores the correlation
among different components of $\sv$ within a section. The optimal AMP algorithm would use
the vector PME of $\sv$ in the following Gaussian vector channel as a denoiser:
\beq
\rv^l = \sqrt{\eta \hat{P}}\sv^l + \zv^l
\label{eq:decoupled}
\eeq 
where $\sv^l\in\RR^{2^J}$ is distributed according to the vector distribution $p(\sv^l)$,
specified in \eqref{eq:multinomial}, the distribution of a single section of $\sv$,
and $\zv^l\in\RR^{2^J}$ is Gaussian iid with $\mathcal{N}(0,1)$ components, independent
of $\sv^l$. 
The vector PME can be expressed as follows
\beq
f_{t,l}^\text{opt} = \frac{1}{Z(\rv^l)}\sum_{\sv^l} \sv^l \exp\left(-\frac{1}{2}\left\|\rv^l - \sqrt{\eta\hat{P}}\sv^l\right\|^2_2\right)p(\sv^l),
\label{eq:vector_pme}
\eeq 
where $\sum_{\sv^l}$ denotes the sum over all possible configurations within one section and $Z(\rv^l)$ is a normalisation
factor.
Note, that the number of possible configurations of $\sv_l$, and therefore the number of terms in the sum
$\sum_{\sv^l}$, grows exponentially in
$J$ and $K_a$ which makes the direct computation of the vector PME infeasible.
Nonetheless, according to \cite{Ber2020}, the MSE of this infeasible version of AMP can be tracked in the
same way as in Theorem \ref{thm:se_amp} but with the vector channel \eqref{eq:decoupled}
instead of the scalar channel and
$\eta\geq0$ is given as the smallest stationary point of the potential
\beq
i^\text{RS}(\eta) = I_{2^J}(\eta\hat{P}) + \frac{2^J}{2\beta}[(\eta-1)\log_2(e)-\log_2(\eta)]
\label{eq:RS-vector-potential}
\eeq 
where $I_{2^J}(\eta\hat{P})$ denotes the mutual information between $\rv^l$ and $\sv^l$
in the Gaussian \emph{vector} channel \eqref{eq:decoupled} and $\beta = 2^JL/n = R_\text{in}2^J/J$
is the aspect ratio of the matrix $\Am$ in \eqref{eq:yAs_plus_z}. Note, that for $J=0$, i.e. blocks of size one, \eqref{eq:RS-vector-potential}
coincides \eqref{eq:RS-scalar-potential}.

An integral part of the proof of the SE equations in \cite{Bay2011,Ber2020} is that 
the intermediate terms $\Am^\top\zv^t + \thetav^t$, appearing in \eqref{eq:amp}, are indeed asymptotically distributed 
as in the decoupled Gaussian channel \eqref{eq:decoupled}.
Therefore, the error statistics of any estimator that is applied
to $\Am^\top\zv^\infty + \thetav^\infty$ after convergence ($t\to\infty$) can be analysed by studying the
same estimator on the Gaussian channel \eqref{eq:decoupled} at the appropriate channel strength $\eta\hat{P}$. 

A crucial point is that the asymptotic AMP performance is given by the smallest stationary point of
the potential function. In particular, the MSE of AMP with the optimal vector PME \eqref{eq:vector_pme} is described
by the smallest stationary point of 
\eqref{eq:RS-vector-potential}, while the MSE of AMP with the scalar PME \eqref{eq:eta_add} is described by 
the smallest stationary point of \eqref{eq:RS-scalar-potential}.
Nonetheless, the potential functions also contains information about the performance of
the optimal MMSE estimate $\hat{\sv} = \EE[\sv|\Am,\yv]$ \cite{Bar2020a}.
Specifically, the asymptotic MMSE is given by the MSE of the PME of $\sv$ in the Gaussian vector channel
\eqref{eq:decoupled}, similar to the analysis of AMP with the vector PME \eqref{eq:vector_pme},
but the effective channel strength $\eta$
is now the \emph{global} minimizer of \eqref{eq:RS-vector-potential}. 
We illustrate the used terminology in \figref{fig:potential_viz}. Note, that
it is possible that the smallest stationary point and the global minimum coincide. In that
case the AMP estimate coincides with the MMSE estimate.
\begin{figure}
    \centering
    \includegraphics[width=0.8\linewidth]{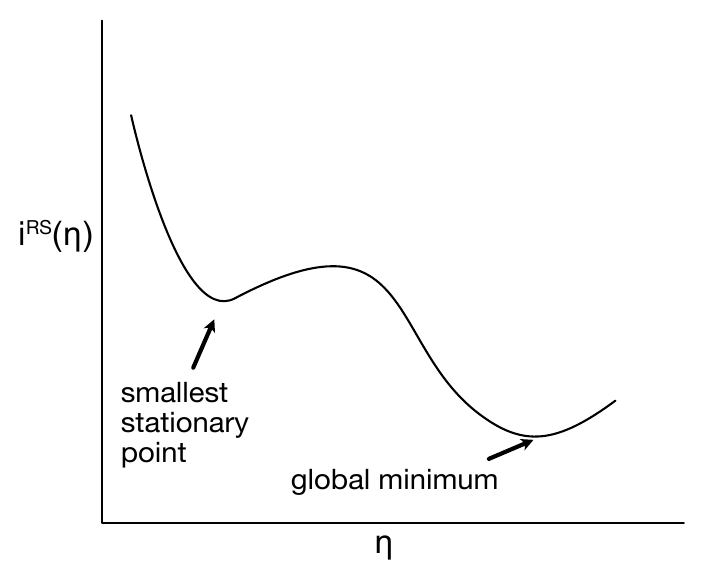}
    \caption{Illustration of the RS-potential
    }
    \label{fig:potential_viz}
\end{figure}
For $\sv$ with iid components, i.e. the special case $J=0$ where \eqref{eq:RS-vector-potential} reduces to \eqref{eq:RS-scalar-potential}, this connection was first discovered, using the non-rigorous replica method,
in \cite{Tan2002} for binary iid signals $\sv$ and generalized to arbitrary iid $\sv$ in
\cite{Guo2005c,Ran2011}. For this reason,
the potential function \eqref{eq:RS-scalar-potential} was termed the \emph{replica symmetric} (RS) potential.
This heuristic discovery was later confirmed rigorously in
\cite{Ree2016,Bar2019}. The proof in \cite[Ch. 4]{Bar2019}
is based on an adaptive interpolation approach and includes the case of block iid signals
which are shown to be described by the global minimizer of \eqref{eq:RS-vector-potential}.

Even in the iid signal case, the original results obtained by the replica method were slightly stronger then the rigorized
versions in \cite{Ree2016, Bar2019, Bar2020a}.
Specifically, they included the so called \emph{decoupling property} \cite{Tan2002,Guo2005c,Guo2009},
which informally states that, in the
case of an iid prior,
the conditional posterior distribution of $\sv$ in the model \eqref{eq:yAs_plus_z}
indeed behaves like the conditional posterior distribution induced by the scalar Gaussian
channel \eqref{eq:decoupled_scalar}.
Motivated by the similarity of the RS potentials in the iid and in the block
iid cases, respectively in \eqref{eq:RS-scalar-potential} and \eqref{eq:RS-vector-potential}, we conjecture that the decoupling
principle holds (subject to the validity of the replica-based analysis) also
for the  block iid case, when the length of the blocks $2^J$ is constant with $n$
and $L$ . This is summarized in the following
\begin{assumption}[Decoupling Property]
    \label{conj:decoupling}
    For an arbitrary section index $l$ let $\sv^l_p$ be a sample from the
    marginal posterior distribution $p_{\sv^l|\yv,\Am}$ where $\yv$ is created
    according to model \eqref{eq:yAs_plus_z} and $\sv^l_0$ being distributed according to
    the prior distribution of one section \eqref{eq:multinomial}.
    We say that the decoupling property holds in some specified limit
    if the joint distribution of $(\sv^l_0,\sv^l_p)$ converges to the joint
    distribution of $(\sv^l_0,\sv^l)$ defined as follows.
    $\sv^l_0$ is again distributed according to the prior
    of a single section while, conditional on $\sv^l_0$,
    $\sv^l$ is distributed according to the posterior distribution $p_{\sv^l|\rv^l}$
    evaluated at the output of a Gaussian vector channel 
    $\rv^l = \sqrt{\eta\hat{P}}\sv^l_0 + \zv^l$ with input $\sv^l_0$.
    The factor $\eta\geq0$ is given as the \emph{global} minimizer of
    \eqref{eq:RS-vector-potential}.
\end{assumption}
\begin{remark}
    The decoupling property implies that the error statistics of every estimator that is
    based purely on the marginal
    posterior distribution, including the SBS-MAP estimator, is asymptotically given by the error statistic of the same estimator but with the
    true posterior distribution replaced by a posterior distribution in a degraded
    Gaussian channel,
    where the degradation
    factor $\eta$ is given as the solution of the minimisation of \eqref{eq:RS-vector-potential}.
    Therefore, asymptotically, once $\eta$ has been determined,
    the error of the SBS-MAP estimator can be computed as the error of
    the estimator $\hat{s}_i = \argmax \PP (s^l_i = s| \rv^l)$ in the channel \eqref{eq:decoupled}.
    Claim \ref{conj:decoupling} is supported by several reasons:
    \begin{enumerate}
        \item It was proven in \cite[Ch. 4]{Bar2019}
            that the mutual information between $\sv$ and $\yv$ in \eqref{eq:yAs_plus_z} converges
            to the RS potential \eqref{eq:RS-vector-potential}.
            Through a limiting argument together
            with the I-MMSE theorem it is possible to conclude the convergence of the posterior-mean
            estimation error of $\sv$ from \eqref{eq:yAs_plus_z}
            to the MMSE in the Gaussian vector channel \cite[Corollary 7]{Bar2020a}.
        \item The replica-based proof of the scalar decoupling property in \cite{Guo2005c} consists of two steps.
            First it is shown that, subject to the validity of the replica analysis, the
            mutual information $I(\yv;\sv)$ in \eqref{eq:yAs_plus_z} converges to \eqref{eq:decoupled_scalar}.
            Then it is shown that all the moments of 
            the joint distribution of $(\sv^l_0,\sv^l)$ converge for an arbitrary index $l$
            to the joint distribution of input
            and posterior output in a degraded Gaussian channel.
            The second step is obtained by essentially \emph{repeating the 
            same calculation} as in the first step.
        \item In \cite{Tan2002} it was argued that the mutual information is equivalent
            to the cumulant generating function and as such it contains the information
            about all the moments of the
            input-output distribution.
            Therefore the convergence of the mutual information would imply the decoupling property. 
            Nonetheless, it is non-trivial to confirm such a claim in the given setting. 
        \item There is an strong connection between the SE of AMP and the replica method.
            As mentioned, the potential \eqref{eq:RS-vector-potential}
            also describes the performance of the optimal, computationally infeasible,
            AMP algorithm.
            The SE of AMP though has been proven rigorously \cite{Ber2020},
            even for the case of block iid input distributions,
            and it is an integral part of the proof of the SE that
            the intermediate terms $\Am^\top\zv + \thetav^t$, appearing in \eqref{eq:amp},
            are asymptotically distributed 
            as in the decoupled Gaussian channel \eqref{eq:decoupled}.
            The major role that SE plays in the rigorous proof of the
            RS-potential in \cite{Bar2020a} leads us to believe that the decoupling property should also
            translate.
    \end{enumerate}
    Nevertheless, the decoupling property, as stated here,
    is not supported by rigorous proofs in the current literature,
    Showing the decoupling property for block iid signals (wether via replica
    analysis or in a fully rigorous way) remains an interesting open problem
    for future work.
\end{remark}

In the remaining chapter we give a mathematical analysis of the vector-channel based potential 
\eqref{eq:RS-vector-potential} and the scalar-channel based potential
\eqref{eq:RS-scalar-potential} and show that the componentwise PME in \eqref{eq:eta_add} is the best
componentwise approximation of the vector PME and furthermore,
in the typical sparse setting, i.e. $K_a \ll 2^J$, the difference between the minima of the
two potential functions \eqref{eq:RS-vector-potential} and \eqref{eq:RS-scalar-potential}
is negligibly small and therefore the global minimizer of the potential \eqref{eq:RS-scalar-potential}
can be used to approximate the global minimizer of \eqref{eq:RS-vector-potential}.
This means that the error probability of both the presented
AMP estimator and, through the decoupling property, the SBS-MAP estimator can be characterized by the local and global minimizers
of the potential function \eqref{eq:RS-scalar-potential} based in the mutual information
in a scalar Gaussian channel.
The connection between \eqref{eq:RS-scalar-potential} and
\eqref{eq:RS-vector-potential} is made precise in the following theorem:
\begin{theorem}
    \label{thm:error_bound}
    Let $\eta_\text{opt}$ be the global minimizer of \eqref{eq:RS-vector-potential} and
    let $\tilde{\eta}_\text{opt}$ be the \emph{global} minimizer of \eqref{eq:RS-scalar-potential}
    then $\eta_\text{opt} > \tilde{\eta}_\text{opt}$ and
    \beq
    \eta_\text{opt} - \tilde{\eta}_\text{opt} = \mathcal{O}\left(\frac{R_\text{in}\log K_a}{\sqrt{J}}\right)
        \label{eq:eta_bound}
    \eeq 
    $\hfill\square$
\end{theorem}
\begin{proof}
Analogous to \eqref{eq:SE} the optimality condition for $\eta_\text{opt}$
can be found by setting the derivative of \eqref{eq:RS-vector-potential}
to zero and using the I-MMSE theorem for a Gaussian vector channel.
This gives the condition
\beq
\eta^{-1} = 1 + \beta\hat{P}\frac{\text{mmse}_{2^J}(\eta\hat{P})}{2^J}
\label{eq:SE_vector}
\eeq 
where $\text{mmse}_{2^J}(\eta\hat{P})$ is the MMSE of estimating $\sv^l$ in the Gaussian vector
channel \eqref{eq:decoupled}.
Let us introduce the mismatched MSE function. For an arbitrary probability distribution
$q: [0:K_a]^{2^J}\to [0,1]$ we define
\beq
\text{mse}_q(t) = \EE\|\sv-\hat{\sv}_q(\sqrt{t}\sv + \Zm,t)\|_2^2
\label{eq:mse_def}
\eeq
with
\beq
\hat{\sv}_q(\rv,t) =
\sum_{\sv \in [0:K_a]^{2^J}} \sv \frac{\exp(-\|\rv - \sqrt{t}\sv\|_2^2/2)q(\sv)}{\sum_{\sv'} \exp(-\|\rv - \sqrt{t}\sv'\|_2^2/2)q(\sv')}
\eeq
The expression in \eqref{eq:mse_def} is the MSE of a (mismatched) PME in a Gaussian vector channel of the form \eqref{eq:decoupled}
with respect to some prior distribution $q(\sv)$,
which may differ form the true prior $p_\sv(\sv)$.
It is clear that
from the minimality of the MMSE function that $\text{mmse}_{2^J}(t)\leq\text{mse}_{q}(t)$
for all $t$ with equality if $q(\sv) = p_\sv(\sv)$.
This means, that calculating the fixed-point of \eqref{eq:SE_vector},
with $\text{mmse}_{2^J}(t)$ replaced by $\text{mse}_q(t)$ for any choice of $q(\sv)$ gives an
upper bound on $\eta_\text{opt}$.
The $L^1$-distance between the functions $\text{mmse}_{2^J}(t)$ and
$\text{mse}_q(t)$ is quantified
by the following result from \cite{Ver2010}
\beq
\begin{split}
&\frac{1}{2}\|\text{mse}_q - \text{mmse}_{2^J}\|_{L^1}\\
&= \frac{1}{2}\int_0^\infty[\text{mse}_q(t) - \text{mmse}_{2^J}(t)]\mathrm{d}t = D(p_\sv\ \|\ q).
    \label{eq:verdu}
\end{split}
\eeq
where $D(p_\sv\ \|\ q)$ denotes the KL-divergence between the distributions $p_\sv$ and $q$.
We focus on product distributions of the form $q(\sv) = \prod_{i=1}^{2^J}q_i(s_i)$, since
for such distributions the vector MSE function \eqref{eq:mse_def} becomes the sum of
scalar MSE functions, which are easy to calculate.
Moreover, a simple calculation in Appendix \ref{appendix:product} shows that the $L^1$
distance \eqref{eq:verdu} is minimized by
the product distribution whose factors $q_i$ match the marginals of $p_\sv$.
We prove in Appendix \ref{appendix:kl} that the KL-divergence between the multinomial
distribution $p_\sv$ and the product distribution of its marginals satisfies
\beq
D\left(p_\sv\ \middle\|\ \prod_{i=1}^{2^J}p_i(s_i)\right) = \mathcal{O}(\log K_a)
\label{eq:kl}
\eeq
By substituting $\eta\hat{P} = t$ we get from \eqref{eq:verdu}:
\beq
\begin{split}
    &\int_0^1 \text{mse}_q(\eta\hat{P}) - \text{mmse}_{2^J}(\eta\hat{P})\mathrm{d}\eta \\
    &= \frac{1}{\hat{P}}\int_0^{\hat{P}}\text{mse}_q(t) - \text{mmse}_{2^J}(t)\mathrm{d}t \\
    &< \frac{D\left(p_\sv\ \middle\|\ \prod_{i=1}^{2^J}p_i(s_i)\right)}{\hat{P}}
\end{split}
\eeq 
Furthermore, since all marginals of $p_\sv$ are identical and given by the binomial distribution, the mismatched
MSE function of $q(\sv) = \prod p_i(s_i)$ takes the form
\beq
\text{mse}_q(t) = 2^J \text{mmse}(t),
\eeq 
where $\text{mmse}(t)$ is the MMSE function in the scalar Gaussian channel \eqref{eq:decoupled_scalar}
that appears in \eqref{eq:SE}. So it follows from \eqref{eq:verdu} and \eqref{eq:kl} that
as a function of $\eta \in [0,1]$
\beq
\left\|\text{mmse}(\cdot\hat{P}) - \frac{\text{mmse}_{2^J}(\cdot\hat{P})}{2^J}\right\|_{L^1} 
= \mathcal{O}\left(\frac{\log K_a}{2^J\hat{P}}\right)
    \label{eq:mmse_L1}
\eeq
Therefore, with $J\to\infty$, the difference between the per-component vector MMSE
function in \eqref{eq:SE_vector} and the scalar MMSE function in \eqref{eq:SE} 
converges exponentially fast in $L^1$ norm to zero.
We show in Appendix \ref{appendix:eta_bound} that exponentially
fast convergence in $L^1$-norm implies exponentially fast pointwise 
convergence everywhere except on a set of arbitrary small measure.

The right-hand sides (RHSs) of \eqref{eq:SE_vector} and \eqref{eq:SE} are multiplied by $\beta\hat{P}$.
With the scaling conditions in Theorem \ref{thm:se_amp} $\beta = R_\text{in}2^J/J$. Therefore, the difference
between the RHS of \eqref{eq:SE_vector} and \eqref{eq:SE} converges pointwise to zero with
rate $\mathcal{O}(R_\text{in}\log K_a/(\delta\sqrt{J}))$ everywhere except on a set of measure $\mathcal{O}(\delta J^{-1/2})$.
Since MMSE functions of Gaussian channels are non-increasing, by a well known theorem in real analysis,
their pointwise convergence implies uniform convergence in all continuity points of the limit functions.
Later, we will show that the common limit function of the mmse functions is continuous on two disjoint intervals
divided by a single discontinuity point.
Therefore, since uniform convergence preserves continuity, the pointwise convergence holds for all $\eta$ except for a single point.
We have established that the difference of the right hand sides of the stationary point equations
\eqref{eq:SE} and \eqref{eq:SE_vector} converges pointwise to zero.
Due to the smoothness of the MMSE functions \cite[Proposition 7]{Guo2011} the convergence carries over to the solutions
of the \eqref{eq:SE} and \eqref{eq:SE_vector} if $\eta_\text{opt}$ and $\tilde{\eta}_\text{opt}$
are away from the discontinuity point of the limiting function, which will follow later from Theorem \ref{thm:limit}.
This concludes the proof of 
\eqref{eq:eta_bound}.
\end{proof}
Theorem \ref{thm:error_bound} allows to calculate
the minima of the potential \eqref{eq:RS-vector-potential} by solving the scalar fixed point equation
\eqref{eq:SE}
numerically. It guarantees that the error will be small, since typically $2^J$
is much larger than $K_a$. Nonetheless, the MMSE function in equation \eqref{eq:SE}
can be further simplified if $J$ grows
large. That is because the coefficients $p_k$ defined in \eqref{eq:binomial},
which can be expressed as
\beq
p_k = p_0\frac{{K_a\choose k}}{(2^J-1)^k}
\eeq 
decay exponentially with $kJ$. This suggests that for large $J$ we can drop all $p_k$
with $k\geq 2$:
\begin{theorem}
    \label{thm:or_mmse}
    Let $\text{mmse}_\text{OR}(t)$ be the MMSE function of estimating the binary variable
    $s_\text{OR}\in \{0,1\}$ in the scalar Gaussian channel
    \beq
        r = \sqrt{t}s_\text{OR} + z
        \label{eq:decoupled_or}
    \eeq 
    where $p(s_\text{OR}=0) = p_0$ and $p(s_\text{OR} = 1) = 1-p_0$ and $z\sim\mathcal{N}(0,1)$
    independent of $s_\text{OR}$. Then $\text{mmse}_\text{OR}(t) \geq \text{mmse}(t)$
    for all $t> 0$ and 
    \beq
    \text{mmse}_\text{OR}(t) - \text{mmse}(t) =
    \mathcal{O}\left(\frac{K_a^2}{2^{2J}}\right)
    \label{eq:diff_mmse_or}
    \eeq 
    $\hfill\square$
\end{theorem}
\begin{proof}
    See Appendix \ref{appendix:or_mmse}.
\end{proof}
We choose the nomenclature \emph{OR} in $\text{mmse}_\text{OR}$ because the distribution of
$s_\text{OR}$ arises as the marginal distribution of $\sv$
if we assume that it is not the sum of the individual messages $\mv_k$
but the OR-sum:
\beq
    \sv = \bigvee_{k=1}^{K_a} \mv_k
\eeq 
Furthermore, let $I_\text{OR}(t)$ be the input-output mutual information
in the channel \eqref{eq:decoupled_or} and 
\beq
i^\text{RS}_{J,\text{OR}}(\eta) = I_\text{OR}(\eta\hat{P}) + \frac{1}{2\beta}[(\eta-1)\log_2(e)-\log_2(\eta)]
\label{eq:RS-OR-potential}
\eeq 
the corresponding RS-potential.
A consequence of Theorems \ref{thm:error_bound} and \ref{thm:or_mmse} is that we can use \eqref{eq:RS-OR-potential}
to find both the global minimizer of \eqref{eq:RS-vector-potential} and the local minimizer
of \eqref{eq:RS-scalar-potential}.
\begin{corollary}
    \label{cor:eta_bound}
    Let $\eta_\text{opt}$ be the global minimizer of \eqref{eq:RS-vector-potential},
    let $\tilde{\eta}_\text{opt}^\text{OR}$ be the global minimizer of \eqref{eq:RS-OR-potential},
    let $\eta_\text{alg}$ be the smallest local minimizer of \eqref{eq:RS-scalar-potential} and
    let $\tilde{\eta}^\text{OR}_\text{alg}$ be the smallest local minimizer of \eqref{eq:RS-OR-potential},
    then $\eta_\text{opt} > \tilde{\eta}_\text{opt}^\text{OR}$
    \beq
    \eta_\text{opt} - \tilde{\eta}_\text{opt}^\text{OR} = \mathcal{O}\left(\frac{R_\text{in}\log K_a}{\sqrt{J}}\right)
        \label{eq:eta_bound_opt}
    \eeq
    Furthermore, if $\eta_\text{opt} \neq \eta_\text{alg}$ then
    $\eta_\text{alg} > \tilde{\eta}_\text{alg}^\text{OR}$ and:
    \beq
    \eta_\text{alg} - \tilde{\eta}_\text{alg}^\text{OR} = \mathcal{O}\left(\frac{R_\text{in}\log K_a}{\sqrt{J}}\right)
        \label{eq:eta_bound_alg}
    \eeq
    $\hfill\square$
\end{corollary}
Note that we have only shown that the difference of the MMSE functions in the channels
\eqref{eq:decoupled} and \eqref{eq:decoupled_or}
converges to zero as $J\to\infty$.
This shows that those functions converge to the same limiting function.
The derivation of this limiting function itself is the subject of the following section.
\subsection{The $J\to\infty$ limit}
\label{sec:limit}
The problem with the numerical evaluation of the stationary points of \eqref{eq:RS-scalar-potential},
even when using Theorem~\ref{thm:or_mmse}, is that
$2^{-J}$ is very small. 
We were able to evaluate the minima of \eqref{eq:RS-scalar-potential} only up to around $J=60$.
\footnote{In a double-precision arithmetic $1-2^{-J}$ evaluates as equal to one at around $J=60$ \cite{IEE2019}.
A more involved implementation with a higher precision may allow the evaluation at higher $J$.}
So even though \eqref{eq:diff_mmse_or}
guarantees that the common $J\to\infty$ limit of the two MMSE functions in the channels 
\eqref{eq:decoupled_or} and \eqref{eq:decoupled} exists, it is not obvious how to calculate it numerically.
To solve these problems we calculate the limit of \eqref{eq:RS-OR-potential}
analytically in a regime where both
$K_a$ and $J$ go to infinity with a fixed ratio $\alpha := J/\log_2 K_a$ for some $\alpha>1$.
The parameter 
$\alpha$ determines the sparsity in the vector $\sv$, e.g. for $K_a = 300$ and $J=15$ we get
$\alpha\sim 1.82$. In this limit $K_a/2^J = K_a^{1-\alpha} \to 0$, i.e. the sparsity
in $\sv$ goes to zero and the error term in Theorem \ref{thm:or_mmse} vanishes.
The RS-potential \eqref{eq:RS-OR-potential} provides a characterisation of the
AMP performance (and by Corollary \ref{cor:eta_bound} and Claim \ref{conj:decoupling}
also of the SBS-MAP performance) when sparsity and the aspect ratio $\beta$ are kept fixed.
The following analysis provides a way of predicting the performance when the aspect ratios are growing
large while at the same time the sparsity becomes small.
We find that a non-trivial limit of the MMSE function exists in the
\emph{energy-efficient} regime, i.e. for $R_\text{in},P \to 0$ with fixed sum-rate
$S_\text{in} = K_a R_\text{in}$ and fixed energy-per-coded-bit
$\mathcal{E}_\text{in} = \SNR/(2R_\text{in})$. Note, that in this limit $R_\text{in} = S/K_a$, i.e.,
the differences in \eqref{eq:eta_bound_opt} and \eqref{eq:eta_bound_alg} go to zero.

\begin{theorem}
    \label{thm:limit}
    In the limit $K_a,J \to \infty$, $R_\text{in},\SNR \to 0$ with fixed 
    $\mathcal{E}_\text{in}$, $S_\text{in}$ and $J = \alpha \log_2 K_a$ for 
    some $\alpha > 1$ the pointwise limit of the RS-potential \eqref{eq:RS-OR-potential}
    is given by
    (up to additive or multiplicative terms that are independent of $\eta$ and therefore do not
    influence the stationary points of $i^\text{RS}(\eta)$):
    \beq
    \begin{split}
        i^\text{RS}_\infty(\eta)&:=\lim_{J\to\infty} i^\text{RS}_{J,\text{OR}}(\eta) =\\
                                &\eta S_\text{in} \mathcal{E}_\text{in}[1-\theta(\eta-\bar{\eta})]\\
                                &+\frac{S_\text{in}}{\log_2 e}\left(1-\frac{1}{\alpha}\right)\theta(\eta-\bar{\eta})\\
                                &+ \frac{1}{2}[(\eta-1)-\ln \eta]
    \end{split}
    \label{eq:rs_limit}
    \eeq
    where
    \beq
        \theta(x):=
        \begin{cases}
            1,\quad \text{if } x > 0\\
            \frac{1}{2},\quad \text{if } x=0\\
            0,\quad \text{if } x<0\\
        \end{cases}
    \label{def:theta}
    \eeq
    and
    \beq
    \bar{\eta} = \frac{1-\frac{1}{\alpha}}{\mathcal{E}_\text{in}\log_2 e}
    \label{def:eta_bar}
    \eeq
    Furthermore, for $\eta \in (0,\bar{\eta})\cup(\bar{\eta},1]$ 
    \beq
    \lim_{J\to\infty}\frac{\mathrm{d}}{\mathrm{d}\eta}i^\text{RS}_\text{J,OR}(\eta)
    = \frac{\mathrm{d}}{\mathrm{d}\eta}i^\text{RS}_\infty(\eta)
    \eeq
    \hfill$\square$
\end{theorem}
\begin{proof}
    See Appendix \ref{appendix:proof_limit}.
\end{proof}
The stationary points of $i^\text{RS}_\infty(\eta)$ can then be calculated analytically,
resulting in simple conditions:
\begin{theorem}
    \label{thm:inner}
    $\eta^* = 1$ is a global minimizer of $i^\text{RS}_\infty(\eta)$, if and only if
    \beq
    S_\text{in}\left(1-\frac{1}{\alpha}\right) < \frac{1}{2}\log_2 (1 + 2S_\text{in}\mathcal{E}_\text{in})
    \label{eq:inner_thm_opt}
    \eeq
    and $\eta^*_\text{loc} = (1+2S_\text{in}\mathcal{E}_\text{in})^{-1}$
    is a local minimizer of $i^\text{RS}_\infty(\eta)$ if and only if 
    \beq
    2S_\text{in} \geq \log_2 e\left(1 - \frac{1}{\alpha}\right)^{-1} - \frac{1}{\mathcal{E}_\text{in}}
    \label{eq:inner_thm_alg}
    \eeq
    \hfill$\square$
\end{theorem}
\begin{proof}
    According to Theorem \ref{thm:limit} the derivative of $i^\text{RS}_\infty(\eta)$ in
    \eqref{eq:rs_limit} is given by
    \beq
    \frac{\partial i^\text{RS}_\infty}{\partial \eta}(\eta) = S_\text{in}\mathcal{E}_\text{in}[1-\theta(\eta-\bar{\eta})] + \frac{1}{2}\left(1-\frac{1}{\eta}\right)
    \eeq
    for $\eta\neq\bar{\eta}$.
    Therefore, the stationary points of $i^\text{RS}_\infty(\eta)$ are
    \beq
        \eta_0^* = (1+2S_\text{in}\mathcal{E}_\text{in})^{-1}
    \eeq
    and
    \beq
        \eta_1^* = 1.
    \eeq
    The first point $\eta_0^*$ is stationary if and only if $\eta_0^* < \bar{\eta}$,
    which,
    after rearranging, gives precisely condition \eqref{eq:inner_thm_alg}.
    Also note, that the second derivative of
    $i^\text{RS}_\infty$ is $(4\eta)^{-2}$, so it is non-negative for all $\eta>0$.
    Therefore the stationary points are indeed minima. A local maximum may appear only at $\eta=\bar{\eta}$
    where $i^\text{RS}_\infty$ is not differentiable.
    The values of $i^\text{RS}_\infty$ at the minimal points are
    \beq
    \begin{split}
    i^\text{RS}_\infty(\eta^*_0) 
    &= \frac{S_\text{in}\mathcal{E}_\text{in}}{1 + 2S_\text{in}\mathcal{E}_\text{in}} + 
    \frac{1}{2}\left[\frac{-2S_\text{in}\mathcal{E}_\text{in}}{1 + 2S_\text{in}\mathcal{E}_\text{in}}+\ln(1+2S_\text{in}\mathcal{E}_\text{in})\right] \\
    &= \frac{\log_2 (1 + 2S_\text{in}\mathcal{E}_\text{in})}{2\log_2 e}
    \end{split}
    \eeq
    if $\eta^*_0< \bar{\eta}$, and
    \beq
    i^\text{RS}_\infty(\eta^*_1) = \frac{S_\text{in}}{\log_2 e}\left(1 - \frac{1}{\alpha}\right)
    \eeq
    It is apparent that $i^\text{RS}_\infty(\eta^*_1)$ is the global minimum if and only if condition
    \eqref{eq:inner_thm_opt} is fulfilled.
    We implicitly used here that $\bar{\eta}\leq 1$,
    that is because condition \eqref{eq:inner_thm_opt} implies $\bar{\eta}< 1$,
    which can be seen by solving inequality \eqref{eq:inner_thm_opt} for $\mathcal{E}_\text{in}$.
\end{proof}
\section{Hard Decision}
\label{sec:decision}
In this section we calculate the symbol detection error probabilities, assuming that the output
of the inner decoder for each position $s_i^l$ is given as the true signal in independent
Gaussian noise with an effective channel
strength $\eta\hat{P}$.
For the AMP algorithm this assumption is well justified as mentioned in
the beginning of Section \ref{sec:inner}
while for the optimal symbol-wise detector this assumption is justified
only by Claim \ref{conj:decoupling}. Note, that the MAP estimate on $s_i^l$ from
a vector observation $\rv^l$ in the channel \eqref{eq:decoupled} is the one which
maximises $p(s_i^l|\rv^l)$. Such an estimate is hard to analyse in general so we resort to analysing
the suboptimal estimator that maximises $p(s_i^l|r^l_i)$. 
This error probability of this estimator is described by the error probability of MAP estimation in
the scalar Gaussian channel \eqref{eq:decoupled_scalar} with the appropriate effective channel strength
$\eta_\text{opt}$. As we will later show, the suboptimal estimation is sufficiently good
in the sparse regime $K_a \ll 2^J$.

The previous section described
in depth how the effective channel strength $\eta\hat{P}$ can be calculated for the AMP and MAP estimation
respectively. For the detection, we
consider only the problem of deciding between $s=0$ and $s\geq1$ from the Gaussian observation $r$,
specified by \eqref{eq:decoupled_scalar}.
We focus only on the support information for two reasons.
First, because the outer code considered in this paper makes use only of the
support information. Second, because we are interested in the typical setting 
 where $K_a \ll 2^J$. The intuition is that, in this regime,
collisions are so rare
that the error-per-component of treating each component as if collisions are impossible
is negligibly small. Let $\hat{\rho}$ be an estimate of $\mathds{1}(s\geq 1)$ given an observation
of $s$ in Gaussian noise according to \eqref{eq:decoupled_scalar}.
We define two types of errors,
the probability of missed detections (Type I errors)
\beq
    p_\text{md} = p(\hat{\rho} = 0 |s \geq 1)
\eeq
and the probability of false alarms (Type II errors)
\beq
    p_\text{fa} = p(\hat{\rho} = 1|s=0).
\eeq
From the Neyman-Pearson Lemma \cite{Poo1994}, the optimal
trade-off between the two types of errors is achieved by choosing $\hat{\rho} = 1$
whenever
\beq
    \frac{p(r|s \geq 1)}{p(r|s = 0)} \geq \theta,
    \label{eq:neyman}
\eeq
where $\theta$ is some appropriately chosen threshold.
If we assume that $s$ takes on only binary values with $p(s>1) = 0$ and $p(s=1) = 1-p_0$,
i.e. the same OR-approximation introduced
in Theorem \ref{thm:or_mmse},
a straightforward calculation shows that by varying $\theta$
the trade-off between $p_\text{md}$ and $p_\text{fa}$ follows the curve defined
by the equation
\beq
Q^{-1}(p_\text{md}) + Q^{-1}(p_\text{fa}) = \sqrt{\eta\hat{P}}
\label{eq:tradeoff}
\eeq 
where $Q(x)$ denotes  the Q-function.
In \figref{fig:neymanp} we plot this curve for $K_a = 300$, $J = 12$ and various values of $\eta\hat{P}$
together with the curves obtained from the precise evaluation of the Neyman-Pearson error probabilities
without the OR-approximation. It is apparent that the difference between those two is barely
recognizable.
By choosing $\theta = p_0/(1-p_0)$ we get the SBS-MAP estimator,
which minimizes the total error probability
\beq
p(\hat{\rho}\neq\rho) = p_\text{fa}p_0 + p_\text{md}(1-p_0).
\eeq
Nonetheless, we find that is useful for practical purposes to vary the threshold
$\theta$ to balance false alarms and missed detections in a way that is adapted to the outer
decoder.
\begin{figure}
    \centering
    \includegraphics[width=\linewidth]{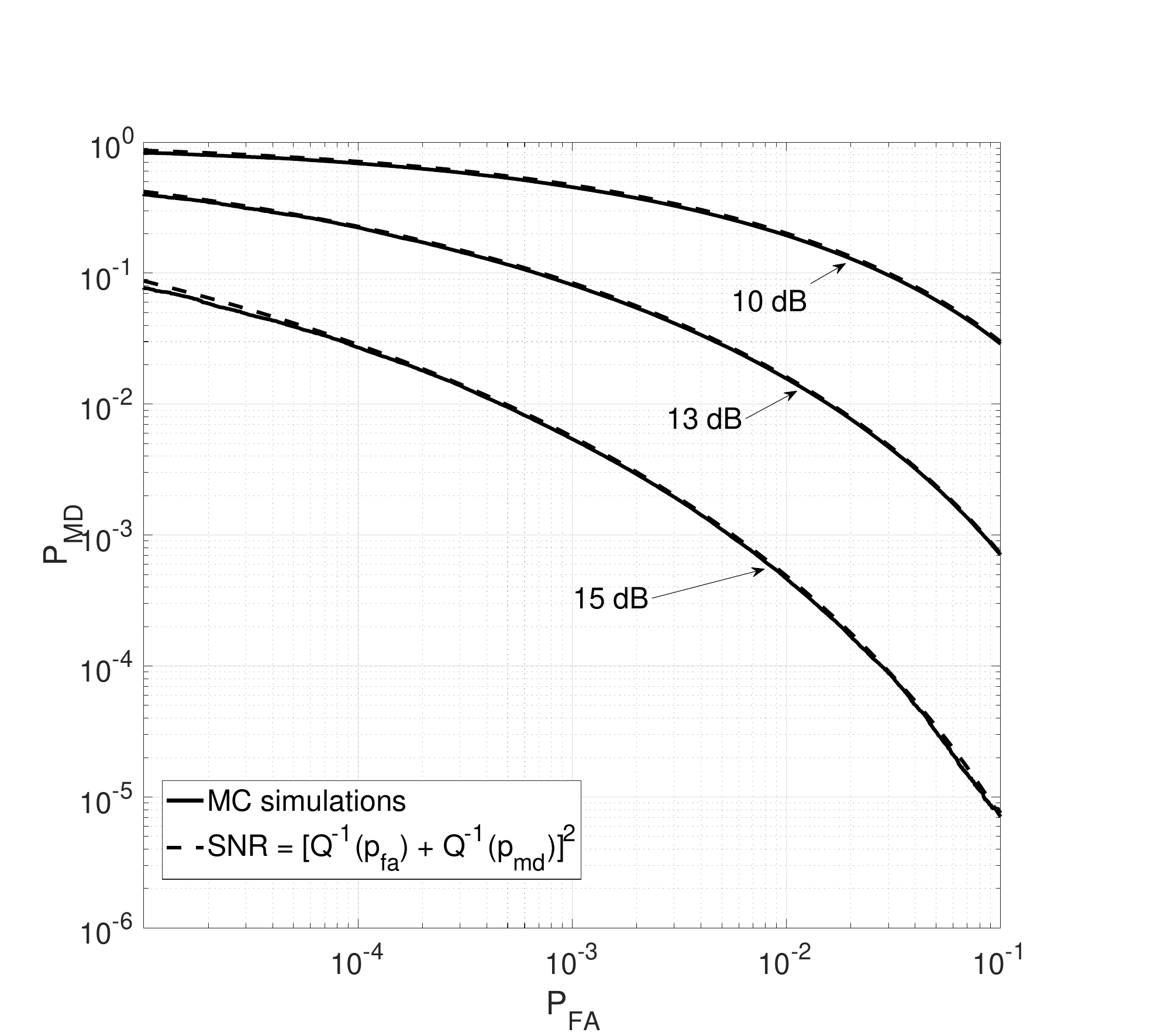}
    \caption{Detection error trade-off for $K_a = 300$ and $J=12$. The solid curves are created by
        Monte-Carlo simulations of the channel \eqref{eq:decoupled_scalar}
        with subsequent likelihood-ratio test, where
        the threshold $\theta$ is varied.
        The dashed lines are calculated by \eqref{eq:tradeoff}. 
    }
    \label{fig:neymanp}
\end{figure}
\section{Outer Channel}
\label{sec:outer}
Let us assume that the inner decoder (SBS-MAP or AMP)
has recovered the support of $\sv = \sum_{k=1}^{K_a} \mv_k$ such that the symbol-wise error probabilities
are given by $p_\text{md}$ and $p_\text{fa}$. We have shown in the previous two sections
how $p_\text{md}$ and $p_\text{fa}$ can be found in the asymptotic limit. The support
of $\sv$ is given by the OR-sum of the messages:
\beq
\rhov = \text{supp}(\sv) = \bigvee_{k=1}^{K_a} \mv_k
\eeq
Let $\hat{\rhov}$ denote the estimated support vector. It can be interpreted as the output
of $L$ uses of a vector OR-MAC
with symbol-wise, asymmetric noise.
We denote the lists of active indices in the $l$-th section by
\begin{equation}
\label{eq:support_lists} 
\Sc_l = \left \{ i \in [1:2^J] : \hat{\rho}^l_i = 1 \right \}. 
\end{equation}
Given $L$ such lists, the outer decoder
is tasked to recover the list of transmitted messages up to permutation. 
An outer code is a subset $\mathcal{C}\subset [1:2^J]^L$ of $|\mathcal{C}| = 2^{JLR_\text{out}}$
codewords. An outer codeword can be equivalently represented as either a set of $L$ indices in
$[1:2^J]$ or as a binary vector in $\RR^{L2^J}$ with a single one in each section of size $2^J$.

Classical code constructions for the OR-MAC, like \cite{Kau1964,Dya1983},
have been focussed on zero-error decoding, 
which does not allow for non-zero per-user-rates as $K_a\to\infty$,
see e.g. \cite{Gyo2008} for a recent survey. 
Capacity bounds for the OR-MAC under the given input constraint have been derived in \cite{Cha1981}
and \cite{Gra1997},
where it was called the ``T-user M-frequency noiseless MAC without
intensity information'' or ``A-channel''. An asynchronous version of this channel was studied in
\cite{Han1996}. Note, that the capacity bounds in the literature are combinatorial and hard to evaluate
numerically for large numbers of $K_a$ and $2^J$.
In the following we will show that, in the typical case of $K_a \ll 2^J$,
a simple upper bound on the achievable rates based on the componentwise entropy is already tight.

From the channel coding theorem for
discrete memoryless channels \cite{El1980} it is known
that a code with per-user-rate $R_\text{out}$ [bits/coded bits]
and an arbitrarily small error probability exists if and only if
\beq
R_\text{out} < \frac{I(\mv_1,...,\mv_{K_a};\hat{\rhov})}{JK_a}.
\eeq
The coding theorem assumes that each user has his own codebook, so the resulting
rate constraint is an upper bound on the achievable rates of an outer code when every user
has the same
codebook. 
The mutual information is:
\begin{align}
    I(\mv_1, ..., \mv_{K_a};\hat{\sv}) = H(\hat{\sv}) - H(\hat{\sv}|\mv_1,...,\mv_{K_a})
\end{align}
where
\beq
H(\hat{\rhov}|\mv_1,...,\mv_{K_a}) = 2^J(p_0\mathcal{H}_2(p_\text{fa}) + (1-p_0)\mathcal{H}_2(p_\text{md}))
\eeq
and $\mathcal{H}_2(\cdot)$ denotes the binary entropy function.
The output entropy $H(\hat{\sv})$
for general asymmetric noise is hard to compute.
A simple upper bound on the entropy of the $2^J$-ary vector OR-channel can be
obtained by the sum of the marginal entropies of $2^J$ independent binary-input binary-output channels.
If we assume the coded messages to be uniformly distributed, i.e.
$\PP (\mv^l_k = \ev_j) = 2^{-J}$ for all $j = [1:2^J]$, then for all $i$
\beq
    p(\rho_i = 1) = 1 - p_0
\eeq 
with $p_0$ given in \eqref{eq:p0} and 
\beq
p(\hat{\rho}_i = 1) = (1-p_0)(1-p_\text{md}) + p_0p_\text{fa}
\eeq
Therefore, after reordering, we get:
\beq
    H(\hat{\sv}) 
    \leq 2^J\mathcal{H}_2((1-p_0)(1 - p_\text{fa} - p_\text{md}) + p_\text{fa})
\label{eq:iid_or_bound}
\eeq
Technically, this is only an upper bound, but we find numerically that it is very
tight and furthermore, in the next section we will show that in the noiseless case it
is actually achievable by an explicit outer code in the familiar limit $K_a,J\to\infty$
with $J=\alpha \log_2 K_a$. To find the limit of \eqref{eq:iid_or_bound} we assume that 
$p_\text{fa} \leq cK_a/2^J = cK_a^{1-\alpha}$
for some constant $c>0$, i.e. the ratio of false positives
to true positives remains at most constant as $K_a,J\to\infty$. An equivalent condition
is 
\beq
\lim_{K_a\to\infty} \frac{\log_2 p_\text{fa}}{\log_2 K_a} \leq 1 - \alpha
\label{eq:error_scaling}
\eeq
If this is not fulfilled the false positives dominate the entropy terms in the mutual information
and the achievable rates go to zero.
We use the fact that for small arguments,
the binary entropy function becomes
\beq
    \mathcal{H}_2(p) \approx p(1-\log_2 p)
\eeq 
and that $(1-p_0) \approx K_a/2^J$. With this, a straightforward calculation shows that
\beq
\lim_{K_a,J\to\infty} \frac{I(\mv_1, ..., \mv_{K_a};\hat{\sv})}{JK_a} 
\leq (1-p_\text{md})\left(1-\frac{1}{\alpha}\right)
\label{eq:outer_bound}
\eeq 
In the following we assume for simplicity
that the inner decoder works error free, i.e. $p_\text{fa} = p_\text{md} = 0$.
Interestingly, the bound $1 - \alpha^{-1}$
is achievable 
by a random code with a \emph{cover decoder}, a construct often used in group testing
literature.
Given $\text{OR}(\mathcal{L})$, the OR-combination of $\mathcal{L}$,
a list of $K_a$ codewords,
the cover decoder goes through the whole codebook and produces a list of codewords that are covered
by $\text{OR}(\mathcal{L})$.
By construction the cover decoder will find all codewords in $\mathcal{L}$
and the error probability is governed by the number of false positives $n_\text{fa}$. We
assume that if the decoder finds more than $K_a$ codewords, it discards exceeding
codewords at random until the list contains only $K_a$ codewords. Therefore the per-user error probability of the cover decoder
is given as
$P_e = n_\text{fa}/(K_a + n_\text{fa})$. We write $\cv_1 \subset \cv_2$ if a binary vector
$\cv_1$ is covered by a binary vector $\cv_2$, that is if for all $i$ with $c_{1,i} = 1$
also $c_{2,i} = 1$.
\begin{theorem}
    \label{thm:random_coding}
    Let $\mathcal{C}$ be an outer codebook of size $2^{LJR_\text{out}}$, where the position of
    each codeword in each section is chosen uniformly at random. Then the error probability
    of the cover decoder vanishes in the limit $L,K_a,J\to\infty$ with $J = \alpha \log_2 K_a$
    for some $\alpha >1$ if
    \beq
        R_\text{out} < 1 - \frac{1}{\alpha}
        \label{eq:random_coding_condition}
    \eeq 
    $\hfill\square$
\end{theorem}
\begin{proof}
    Let $\mathcal{L}$ be a list of $K_a$ arbitrary codewords from $\mathcal{C}$. Then
    \begin{align}
        P(n_\text{fa}\geq 1) 
        &= P\left(\bigcup_{\cv\notin \mathcal{L}}\{\cv\subset\text{OR}(\mathcal{L})\}\right)\\
        &\leq \sum_{\cv\notin\mathcal{L}} P\left(\cv\subset\text{OR}\left(\mathcal{L}\right)\right) \\
        &\leq 2^{LJR_\text{out}}\max_{\cv\notin \mathcal{L}} P\left(\cv\subset\text{OR}\left(\mathcal{L}\right)\right) \\
        &= 2^{LJR_\text{out}}\max_{\cv\notin \mathcal{L}} \prod_{l=1}^L P\left(\cv^l \subset\text{OR}\left(\mathcal{L}^l\right)\right) \\
        &= 2^{LJR_\text{out}} \left(\frac{K_a}{2^J}\right)^L\\
        &= 2^{LJR_\text{out} + L(1-\alpha)\log_2 K_a}\\
        &= 2^{LJ(R_\text{out}-(1-\alpha^{-1}))}
    \end{align}
    In the second line we have used the union bound. In the third line
    the non-negative sum is upper bounded by its maximum term times the number of
    codewords not in $\mathcal{L}$ which is $2^{LJR_\text{out}} - K_a \leq 2^{LJR_\text{out}}$.
    In the fourth line we have
    used that the entries of each section are chosen independently of each other.
    In the fifth line $\text{OR}(\mathcal{L}^l)$ denotes the OR-combination of the $l$-th section
    of the codewords in $\mathcal{L}$. The probability that a random
    number from $[1:2^J]$ is contained in a fixed set of size $K_a$ is given by
    $1 - (1-2^{-J})^{K_a}$, which becomes $K_a/2^J$ for small $K_a/2^J$. This probability
    is the same for all codewords not in $\mathcal{L}$ which allows to drop the maximum.
    It is apparent that the error probability vanishes for any $L$ and $J\to\infty$
    if condition \eqref{eq:random_coding_condition} is fulfilled.
\end{proof}
\begin{remark}
    The proof of Theorem \ref{thm:random_coding} can easily be extended to include false positives. 
    For that we introduce modified lists $\tilde{\text{OR}}(\mathcal{L}^l)$, which
    in addition to the list of transmitted symbols in section $l$ also contain $n_\text{fa}$
    random erroneous entries. If we assume that $n_\text{fa} = cK_a$ for some constant 
    $c>0$ the result of Theorem \ref{thm:random_coding} is unchanged. Since
    $p_\text{fa} = n_\text{fa}/2^J$, this condition is equivalent to \eqref{eq:error_scaling}.
\end{remark}
We can also derive a finite length upper bound on the achievable outer rates with the cover decoder
in a more direct, combinatorial way.
\begin{theorem}
    Any outer code that can guarantee error-free recovery under cover decoding for $K_a$ users with
    $L$-sections in the limit $K_a,J\to\infty$ with $J=\alpha\log_2 K_a$ has to satisfy:
    \beq
    R_\text{out} \leq 1 - \frac{1}{\alpha} + \frac{1}{\alpha L}
    \eeq 
    $\hfill\square$
\end{theorem}
\begin{proof}
    We first show that any error free code has to satisfy
    \beq
        \frac{2^{LJR_\text{out}}}{K_a} \leq \frac{2^{JL}}{K_a^{L}}
        \label{eq:packing}
    \eeq 
    To see this, assume an outer code that is error-free, i.e. for any list of $K_a$ codewords
    the OR-combination of theses codewords does not cover any other codeword
    that is not in the list. Then any two non-intersecting lists, $\mathcal{L}_1$
    and $\mathcal{L}_2$, of $K_a$ codewords create
    two non-intersecting lists of $K_a^L$ possible sensewords.
    To see that they are non-intersecting,
    note, that due to the error-free property none of the codewords in
    $\mathcal{L}_2$ is covered by the OR-combination
    of all codewords from $\mathcal{L}_1$, which we denote by $\text{OR}(\mathcal{L}_1)$.
    This means that each codeword from $\mathcal{L}_2$
    differs from $\text{OR}(\mathcal{L}_1)$ in at least one position.
    But this also means that any OR-combination of codewords from $\mathcal{L}_2$ differs
    from $\text{OR}(\mathcal{L}_1)$ in at least one position. \\
    Now divide the set of all $2^{LJR_\text{out}}$ codewords into distinct lists of length
    $K_a$, then each of these lists creates a distinct list of $K_a^L$ sensewords, whose
    total number has to be limited by the size of the space:
    \beq
        \frac{2^{LJR_\text{out}}}{K_a}K_a^L \leq 2^{JL}
    \eeq 
    This is precisely the statement of \eqref{eq:packing}. If we use the scaling condition
    $2^J = K_a^{\alpha}$ and take the limit $K_a,J\to\infty$ we get the statement of the
    theorem.
\end{proof}
For $L=1$ we get that
$R_\text{out} \leq 1$, which can obviously be achieved, since for a single section no outer
code is necessary. 
\subsection{Tree code}
The first practical coding scheme for the outer OR-MAC with the sectionized structure
has been presented in \cite{Ama2020a}. It works as follows:
The $B$-bit message is divided into blocks of size 
$b_1, b_2, \ldots, b_L$ such that $\sum_l b_l = B$ and such that
$b_1 =  J$ and $b_l < J$ for all $l = 2, \ldots, L$. 
Each subblock $l = 2, 3, \ldots, L$ is augmented to size $J$ by appending
$\pi_l = J - b_l$ parity bits,  
obtained using pseudo-random linear combinations of the information bits of the
previous blocks $l' < l$.
Therefore, there is a one-to-one association between the set of all sequences of coded blocks and
the paths of a tree of depth $L$.
The pseudo-random parity-check equations generating the parity bits 
are identical for all users, i.e., each user makes use exactly of the same outer
{\em tree code}.
For more details on the outer coding scheme, please refer to  \cite{Ama2020a}. 

Let $\Sc_l,\ l=1,...,L$
be the list of active indices in the $l$-th section, defined in \eqref{eq:support_lists}. 
Since the sections contain parity bits with parity profile $\{0,\pi_2, \ldots, \pi_L\}$, 
not all message sequences in $\Sc_1 \times \Sc_2 \times \cdots \times \Sc_L$ are possible. 
The role of the outer decoder is to identify all possible message sequences, i.e., those corresponding to
paths in the tree of the outer tree code \cite{Ama2020a}.  
The output list $\Lc$ is initialized as an empty list. Starting from $l=1$ and proceeding in order, the decoder converts the integer indices 
$\Sc_l$ back to their binary representation, separates data and parity bits, computes the parity 
checks for all the combinations with messages from the list $\Lc$ and extends only the paths 
in the tree which fulfill the parity checks. A precise analysis of the error probability
in various  asymptotic regimes is given in \cite{Ama2020a}.
Specifically, the analysis
shows that the error probability of the outer code goes to zero
in the limit $K_a,J\to \infty$ with $J = \alpha \log_2 K_a$ and some $\alpha>1$
\footnote{We deviate slightly from the notation in \cite{Ama2020a}, where the scaling parameter
$\alpha'$ is defined by
$B = \alpha'\log_2 K_a$ and the number of subslots is considered to be constant.
It is apparent that
those definitions are connected by $\alpha' = LR_\text{out}\alpha$.}
if the total number of parity bits $P = \sum_{l=2}^L \pi_l$ is chosen as (\cite[Theorem 8 and 9]{Ama2020a})
\begin{enumerate}
    \item $P = (L+\delta - 1)\log_2 K_a$ for some constant $\delta>0$ if all the parity bits
        are allocated in the last slots.
    \item $P = c(L-1)\log_2 K_a$ for some constant $c>1$ if the parity bits are allocated evenly
        at the end of each subslot except for the first.
\end{enumerate}
In the first case the complexity scales like $\mathcal{O}(K_a^{R_\text{out}L}\log K_a)$,
since there is no pruning in the first $R_\text{out}L$ subslots,
while in the second case the complexity scales linearly with $L$
like $\mathcal{O}(LK_a\log K_a)$. The corresponding outer rates are
\beq
\begin{split}
    R_\text{out} &= B/(B+P)\\
                 &= 1 - P/(B+P)\\
                 &= 1 - P/(LJ)\\
                 &= 1 - \frac{L+\delta - 1}{L\alpha}\\
                 &= 1 - \frac{1}{\alpha}  + \frac{1}{L}\frac{\delta-1}{\alpha}
\end{split}
\eeq
for the case of all parity bits in the last sections and
\beq
\begin{split}
R_\text{out} &= 1 - \frac{c(L-1)}{L\alpha}\\
             &= 1 - \frac{c}{\alpha}  - \frac{c}{L\alpha}
\end{split}
\eeq
for the case of equally distributed parity bits.
In the limit $L\to\infty$ the achievable rates are
therefore $R_\text{out} = 1 - 1/\alpha$ and $R_\text{out} = 1 - c/\alpha$ respectively,
which coincides with the asymptotic upper bound \eqref{eq:outer_bound} for $p_\text{md} = 0$ (up to a constant
for the second case). 

In \figref{fig:outer_rates_noiseless} and \figref{fig:outer_rates_alpha}
we compare empirical simulations with the developed theory.
For the empirical results we fix $B= 100$ bits and $L=8$. Furthermore, for various values
of $J$ and $K_a$ we increase the number of parity bits until a per-user error probability
$P_e < 0.05$ is reached. In practice we use a mixture of the two types of parity profile
described above. We choose the last section to be only parity bits, while the remaining parity
bits are distributed uniformly over the sections $2,..,L-1$. The entropy bounds are calculated
according to \eqref{eq:iid_or_bound} with $p_\text{fa} = p_\text{md} = 0$.
In \figref{fig:outer_rates_alpha} we plot the results
as a function of $\alpha = J/\log_2 K_a$. We can see that the achievable rates of
the tree code with the used parity profile are very well described by the formula
$R_\text{out} = 1 - \alpha^{-1}$ for different $J$ and $K_a$. Furthermore, we can see that
the line $R_\text{out} = 1 - \alpha^{-1}$ is approached from above by the upper bound 
\eqref{eq:iid_or_bound} as $J\to\infty$.
\begin{figure}
    \centering
    \includegraphics[width=\linewidth]{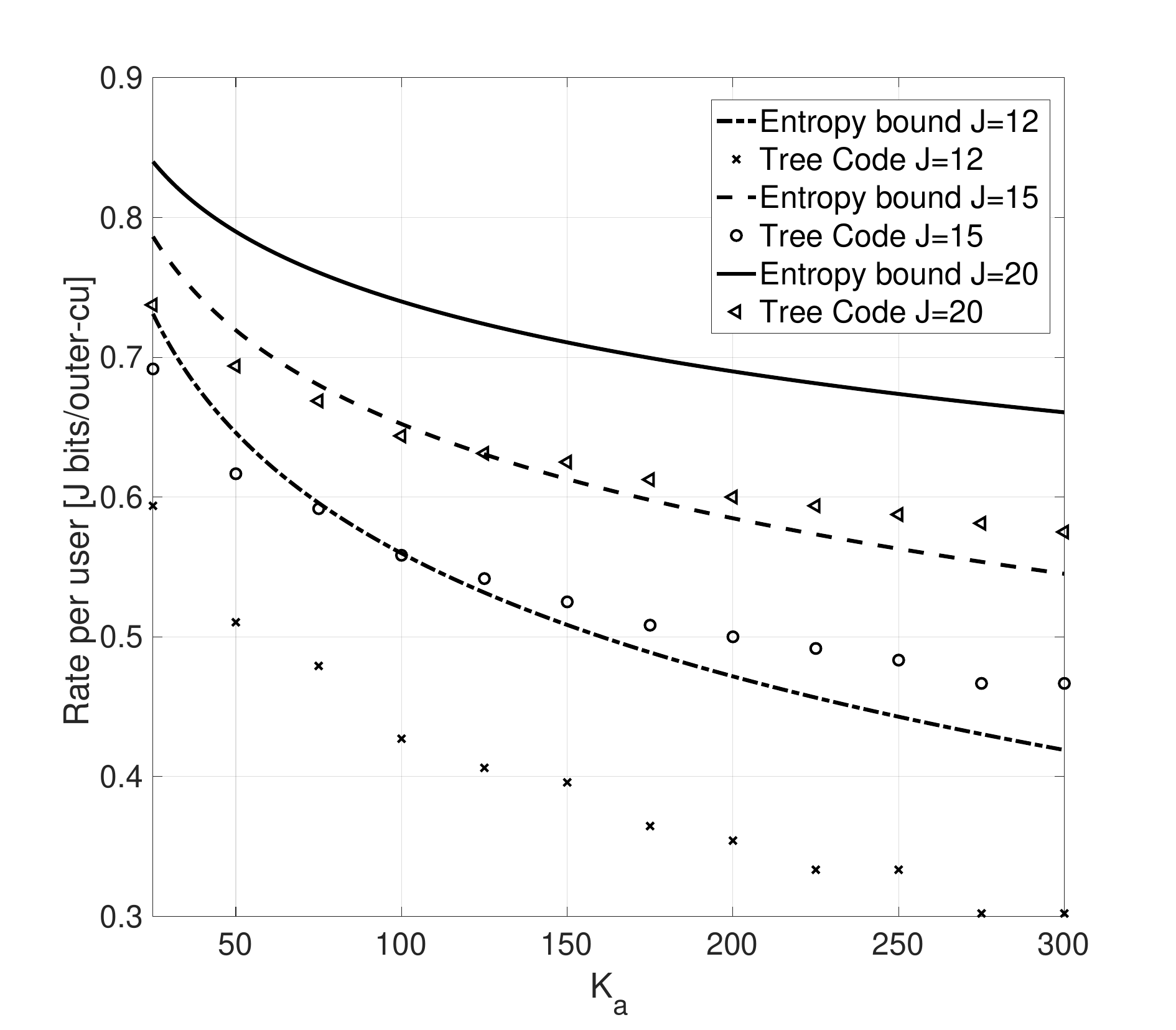}
    \caption{ Achievable rates with the tree code with $L=8$ and $B=100$ bits together
        with the upper bound \eqref{eq:iid_or_bound}.
    }
    \label{fig:outer_rates_noiseless}
\end{figure}
\begin{figure}
    \centering
    \includegraphics[width=\linewidth]{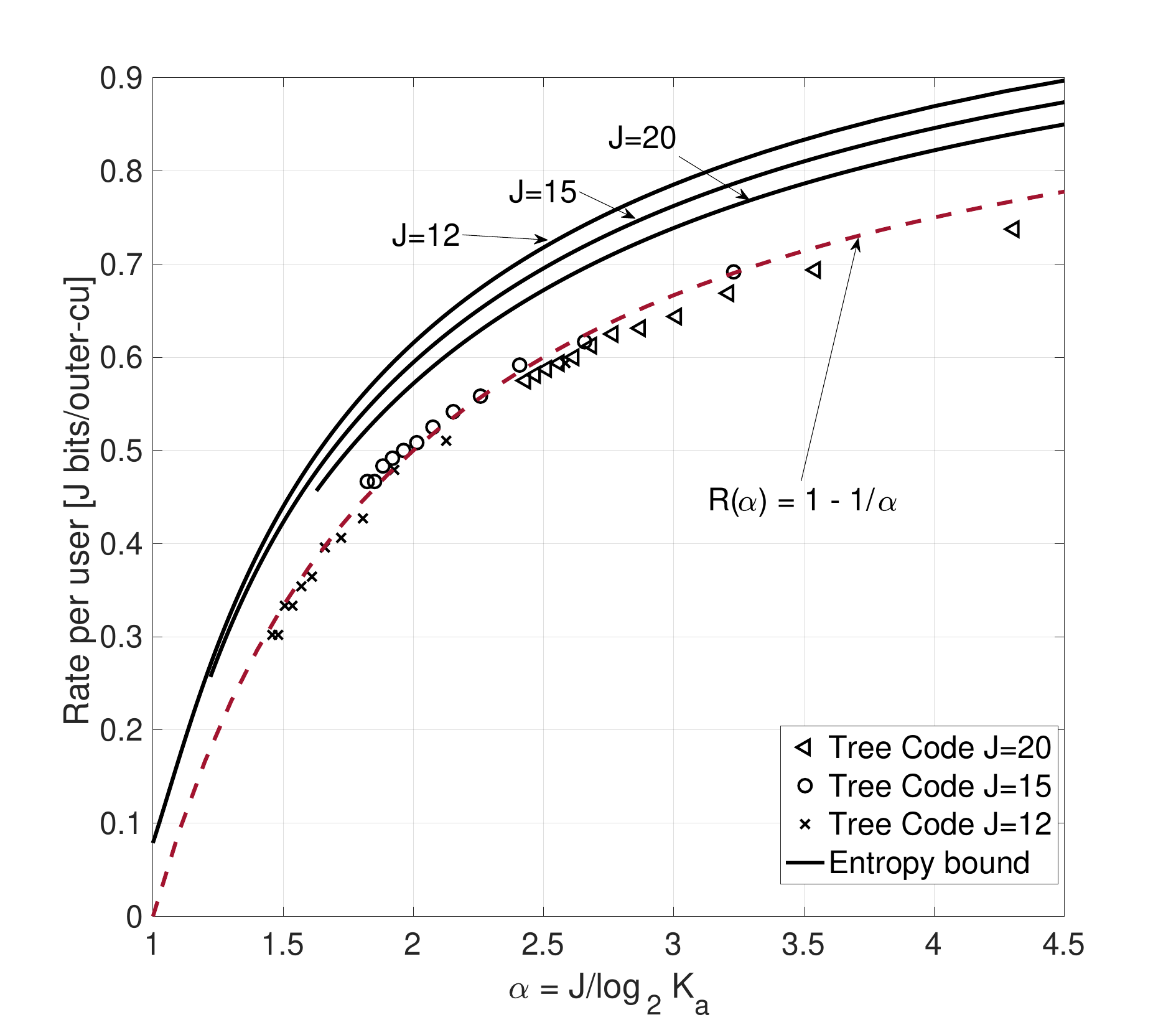}
    \caption{ Achievable rates with the tree code with $L=8$ and $B=100$ bits together
        with the upper bound \eqref{eq:iid_or_bound} as a function of $\alpha$.
    }
    \label{fig:outer_rates_alpha}
\end{figure}
\section{Analysis of the concatenated scheme}
\label{sec:conc}
Let us reformulate Theorem \ref{thm:inner} in terms of the parameters of the concatenated code.
Then the sum-rate is given as  
\beq
S = K_a R = K_a R_\text{in}R_\text{out} = S_\text{in}R_\text{out}
\eeq 
and similarly we have 
\beq
\frac{E_b}{N_0} = \frac{P}{RN_0} = \frac{\mathcal{E}_\text{in}}{R_\text{out}}.
\eeq 
As we have shown in the previous section,
the best achievable outer rate is $R_\text{out} = 1 - \alpha^{-1}$, which 
turns out to coincide with the factor appearing in Theorem \ref{thm:inner}. 
Since the channel strength in the inner channel is given by
$\eta\hat{P}$ and $\hat{P} = n\SNR/L = J\SNR/R_\text{in} = 2J\mathcal{E}_\text{in}$,
the channel strength goes to infinity with $J$ and the error probabilities $p_\text{fa}$ and
$p_\text{md}$ in the outer channel
vanish according to \eqref{eq:tradeoff}.
An important condition to get the asymptotic limit
$R_\text{out} = (1-p_\text{md})(1 - \alpha^{-1})$ in \eqref{eq:outer_bound} is \eqref{eq:error_scaling},
i.e. that the ratio of
false positives to true positives remains at most constant.
This condition requires that
the channel strength in the inner channel has to grow fast enough to ensure that 
the probability of false alarms vanishes faster than the sparsity. In the following 
Theorem we show that the scaling condition for the error probability puts
a constraint on the factor $\eta$.
\begin{theorem}
    \label{thm:error_scaling}
    Let $\eta\hat{P}$ be the channel strength in the scalar Gaussian channel
    \eqref{eq:decoupled_scalar}. In the limit $K_a,J\to\infty$ with $J = \alpha \log_2 K_a$
    for some $\alpha >1$
    the condition \eqref{eq:error_scaling} is fulfilled if and only if 
    \beq
        \eta \geq \bar{\eta}
    \eeq 
    where $\bar{\eta}$ is given in \eqref{def:eta_bar}.
    $\hfill\square$
\end{theorem}
\begin{proof}
    We show only the direction
    \beq
        \eta \geq \bar{\eta} \implies
        \lim_{K_a\to\infty} \frac{\log_2 p_\text{fa}}{\log_2 K_a} \leq 1 - \alpha
    \eeq 
    the reverse implication can be shown similarly.\\
    We can choose the points $p_\text{md}$ and $p_\text{fa}$ on the curve defined by
    \eqref{eq:tradeoff} in a way that $Q^{-1}(p_\text{md})<\epsilon$ for some constant $\epsilon >0$.
    Therefore, for $\hat{P} = 2J\mathcal{E}_\text{in}$
    \beq
        Q^{-1}(p_\text{fa}) \geq \sqrt{\eta\hat{P}} - \epsilon 
        = \sqrt{2\eta J \mathcal{E}_\text{in}} - \epsilon 
    \eeq 
    so
    \beq
    p_\text{fa} \leq Q(\sqrt{2\eta J \mathcal{E}_\text{in}} - \epsilon)
    \eeq 
    and
    \begin{align}
        \lim_{K_a\to\infty}\frac{\log_2 p_\text{fa}}{\log_2 K_a}
        &\leq
        \lim_{K_a\to\infty}\frac{1}{\log_2(e)}\left(\frac{-\eta J\mathcal{E}_\text{in}}{\log_2 K_a}
        + \frac{\mathcal{O}(\sqrt{J})}{\log_2 K_a}\right)\\
        &= -\frac{\eta \alpha \mathcal{E}_\text{in}}{\log_2(e)}
        \label{eq:lim_log_p_fa}
    \end{align}
    where the first line follows from the standard bound
    on the Q-function $Q(x) \leq (2\pi)^{-1/2}\exp(-x^2/2)/x$.
    By reordering, we can see that $\eqref{eq:lim_log_p_fa}<1-\alpha$ if 
    \beq
    \eta > \frac{1-\alpha^{-1}}{\log_2(e)\mathcal{E}_\text{in}} = \bar{\eta}
    \eeq
\end{proof}
The consequences of Theorem \ref{thm:error_scaling} for the concatenated code
are summarized in the following Corollary.
\begin{corollary}
    \label{cor:conc}
    Let $n,L,J,K_a \to \infty$ and $R,\SNR \to 0$ with fixed $E_b/N_0 = \SNR/(2R)$, $S=K_aR$
    and $J = \alpha \log_2 K_a$ for any $\alpha >1$. In this limit there is an outer
    code such that the concatenated code described in Section \ref{sec:coding} with a 
    random Gaussian codebook
    can be decoded with the AMP algorithm \eqref{eq:amp} and $P_e \to 0$
    if and only if
    \beq
    S < \frac{1}{2}\left(\log_2 e - (E_b/N_0)^{-1}\right)
        \label{eq:conc_alg}
    \eeq 
    If the SBS-MAP estimator is used as inner decoder and the decoupling property holds true 
    (Claim~\ref{conj:decoupling}), reliable decoding is possible
    if and only if
    \beq
        S < \frac{1}{2}\log_2 (1 + K_a\SNR)
        \label{eq:conc_cor}
    \eeq
    \hfill$\square$
\end{corollary}
\begin{proof}
    The statement follows immediately from Theorems \ref{thm:inner} and \ref{thm:error_scaling} together with
    the relation $\eta_\text{loc}^* < \bar{\eta}\leq 1$, which is discussed in the proof 
    of Theorem \ref{thm:inner}.
\end{proof}
\begin{remark}
    \normalfont
    In the case $K_a = 1$ no outer code is necessary, so $R_\text{in} = R$ and furthermore
    $S_\text{in} = R$ and $2S_\text{in}\mathcal{E}_\text{in} = \SNR$.
    Hence, if $K_a=1$ is fixed and $J\to\infty$,  which corresponds to
    $\alpha \to \infty$, then Corollary \ref{cor:conc} recovers the statements of
    \cite{Jos2012,Rus2017,Bar2017a}, i.e. that SPARCs are reliable at rates up to the Shannon
    capacity $0.5\log_2 (1 + \SNR)$
    under optimal decoding. Also the algorithmic threshold \eqref{eq:conc_alg}
    coincides with the
    result of \cite{Bar2017a}.
    In that sense Theorem \ref{thm:inner} and Corollary \ref{cor:conc} are
    an extension of \cite{Bar2017a} and show that SPARCs can achieve the optimal rate 
    limit in the unsourced random access scenario. However, notice that the concept of
    our proof technique is simpler, since we make use of Theorem \ref{thm:or_mmse}, which states
    that not only the sections are described by a decoupled channel model, but in the limit
    $J\to\infty$ also the
    individual components. So the result of Theorem \ref{thm:inner}
    can be derived from the stationary points of a simple scalar-to-scalar function.
\end{remark}
\begin{remark}
    In general, most classical multiple-access variants on the AWGN,
    where all the users are assumed to have their own codebook,
    can be represented as sparse recovery problems
    like \eqref{eq:inner_channel}.
    For that, let
    $K_a = 1$ and identify the number of section with the number of users. 
    The matrices $\Am_1,...,\Am_L$ are then the codebooks
    of the individual users and $P_l$ are the transmit power coefficients of different users:
    \begin{itemize}
        \item Fixed $L$ in the limit $J,n\to\infty$ describes the classical
            AWGN Adder-MAC from \cite{El1980}, where each user has his own codebook.
        \item $L,J,n\to\infty$, where only a fraction of the sections are non-zero
            describes the many-access channel treated in \cite{Che2017}
        \item $J$ fixed and $L,n\to\infty$ describes specific version of
            the many-access MAC treated in \cite{Zad2019,Pol2017}
    \end{itemize}
    It is interesting, that in the first case Theorem \ref{thm:inner} gives the
    correct result, after letting $\alpha\to\infty$, $K_a = 1$ and $L=K$.
    The case of $J,n\to\infty$ at finite $L$ is not directly covered though by our analysis
    framework. 
    Nonetheless, our empirical results (e.g. \figref{fig:SE_alg}) show a good
    agreement with the SE predictions even for small $L$. 
    Especially the case $L=1$ is interesting
    since it resembles the U-RA formulation with random coding, as it was already noted in \cite{Zad2019}.
\end{remark}
In \figref{fig:SE_opt} and \figref{fig:SE_alg} we visualize the results of Theorem
\ref{thm:inner}. For that we fix $\alpha = 2$.
For various values of $J$ we set $K_a$ such that $J = \alpha \log_2 K_a$.
For each value of $R_\text{in}$ we then calculate 
$\eta_\text{opt}$ and $\eta_\text{alg}$ using
the approximations of Theorem \ref{thm:error_bound} and Corollary \ref{cor:eta_bound}. 
We repeat this process with increasing $\SNR$
until $\eta_\text{opt}\hat{P}$ and $\eta_\text{alg}\hat{P}$ resp. reach a value of
$(Q^{-1}(p_\text{md})+Q^{-1}(p_\text{fa}))^2$
where the error probabilities are chosen as $p_\text{md} = 0.05/L$, with $L=8$, and $p_\text{md} = 0.01K_a/2^J$.
These are the solid lines in \figref{fig:SE_opt} and \figref{fig:SE_alg}.
The dashed lines are the threshold lines from Theorem \ref{thm:inner}.
Additionally, in \figref{fig:SE_alg} we plot empirical results,
obtained by Monte-Carlo simulations with $L=8$,
where the inner channel with the AMP decoder is simulated with increasing $\SNR$ until
the error probabilities satisfy
$p_\text{md}<0.05/L$ and $p_\text{md}<0.01K_a/2^J$, matching the values above.
We can observe several interesting effects. We can see that the asymptotic
trade-off curve $S_\text{in}(1-\alpha^{-1}) = 0.5\log_2(1+2S_\text{in}\mathcal{E}_\text{in})$
is approached from below by the curves for finite $J$. Also the finite length curves
exhibit a region for small $S_\text{in}$, where $\mathcal{E}_\text{in}$ stays almost constant
up to some value of $S_\text{in}$, and then it starts to grow linearly. Such a behavior
was also observed in e.g. \cite{Zad2019} in the context of finite-blocklength multiple access
on the AWGN channel.
This constant
regions becomes smaller with increasing $J$ and disappears completely in the asymptotic
limit.
We can also see that
there is a region of $S_\text{in}$ in that the algorithmic curve stays almost constant
and matches the optimal curve. That is the region where there is only one unique minimum
in the RS-potential. The empirical simulations in \figref{fig:SE_opt} confirm the qualitative
behavior of the calculated curves. The required energy stays constant over a large region of
$S_\text{in}$ until some point, where it start to grow rapidly.
Note, that Theorem \ref{thm:se_amp} assumes infinite $L$,
but nevertheless the theoretical results match the empirical
simulations with $L=8$ very precisely. 
According to Theorem \ref{thm:inner} in the limit $J\to\infty$
the required energy grows to infinity as
$S_\text{in}$ approaches $\log_2(e)/[2(1-1/\alpha)]$. We can observe that this value
is increased for finite $J$.
In general, we can observe that the asymptotic algorithmic limit is very slowly
approached from \emph{above}. That has the interesting consequence, that for each 
$S_\text{in}$ there is a value $J^*$ below which the required energy decreases with
$J$ and above which the required energy starts to increase again
(see \figref{fig:SE_J} where this is visualized for $S_\text{in} = 2$ with the same parameters
as above).
This reveals a limitation of the AMP algorithm, when used for sparse recovery,
that only manifests when the size of the sections grows large while the inner rate is
held constant. This goes against the intuition that the AMP performance should get better with
increasing sections sizes because random fluctuations are averaged out. 
The sub-optimality of AMP and other iterative algorithms like belief-propagation has been
noted several times, e.g. \cite{Krz2012a,Rus2017,Hsi2018a,Bar2020a}.
The given analysis quantifies the sub-optimality from a different point of view, i.e.
in the information theoretically interesting
regime of constant inner rate. 
\begin{figure}
    \centering
    \includegraphics[width=\linewidth]{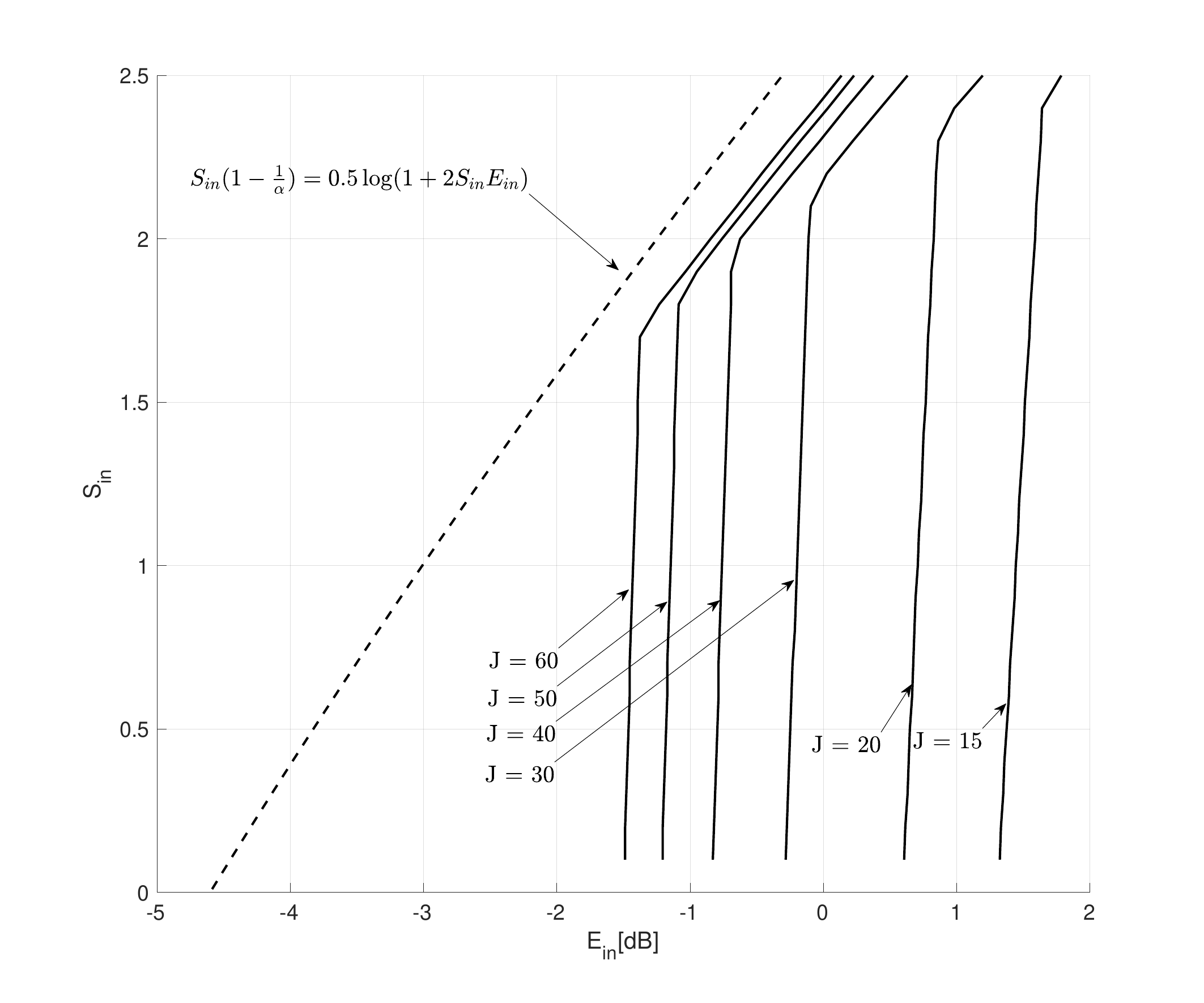}
    \caption{ Required $\mathcal{E}_\text{in}$ to reach specific target error probabilities
        under optimal decoding and $J = \alpha \log_2 K_a$ for $\alpha = 2$
        according to Claim \ref{conj:decoupling} and Theorem \ref{thm:error_bound}.
        The dashed line is the asymptotic limit according to Theorem \ref{thm:inner}.
    }
    \label{fig:SE_opt}
\end{figure}
\begin{figure}
    \centering
    \includegraphics[width=\linewidth]{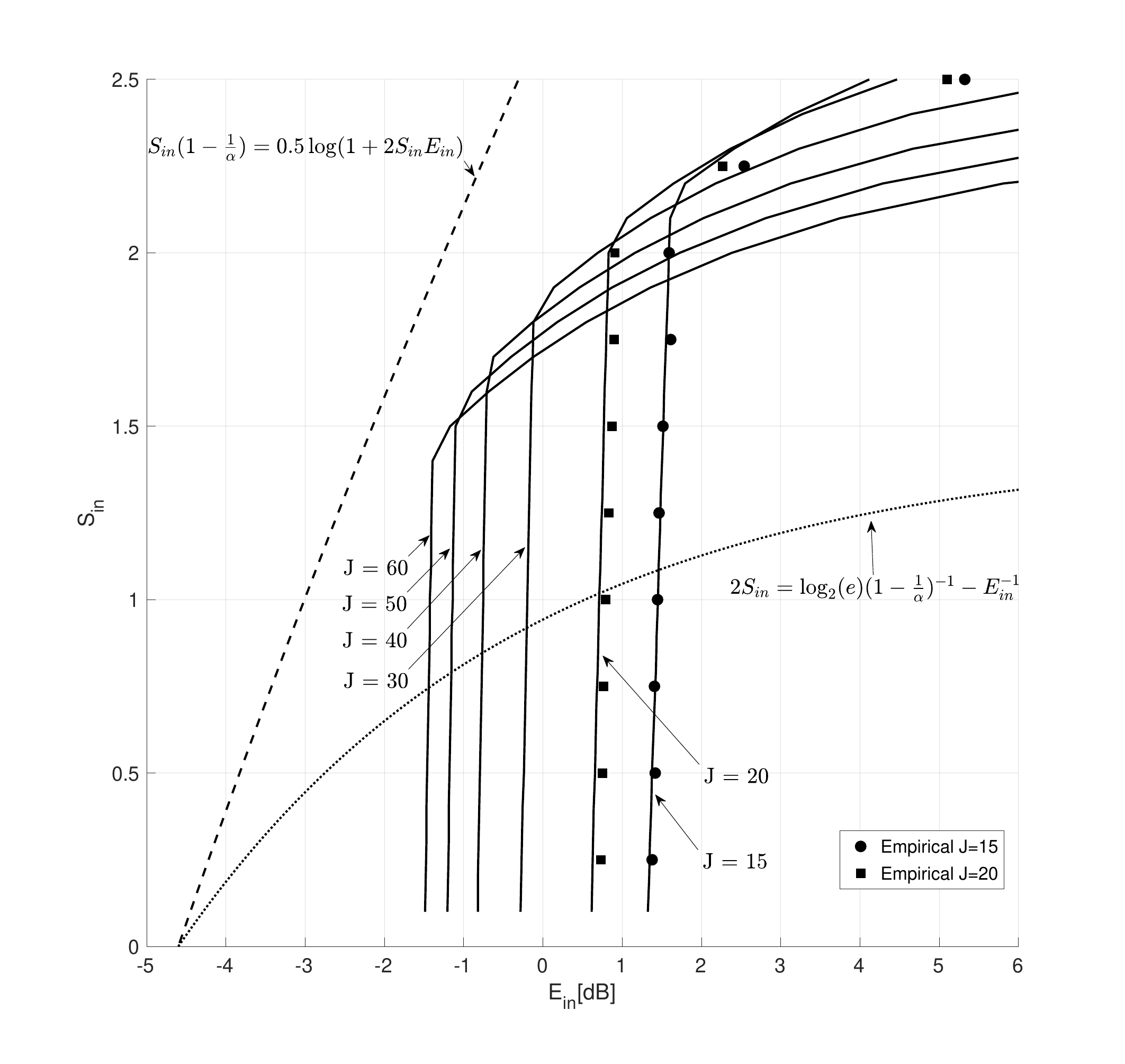}
    \caption{ Required $\mathcal{E}_\text{in}$ to reach specific target error probabilities
        with AMP decoding and $J = \alpha \log_2 K_a$ for $\alpha = 2$
        according to Theorem \ref{thm:se_amp}.
        The empirical simulations
        were conducted with $L=8$.
        The dashed lines are the asymptotic limits according to Theorem \ref{thm:inner}.
    }
    \label{fig:SE_alg}
\end{figure}
\begin{figure}
    \centering
    \includegraphics[width=\linewidth]{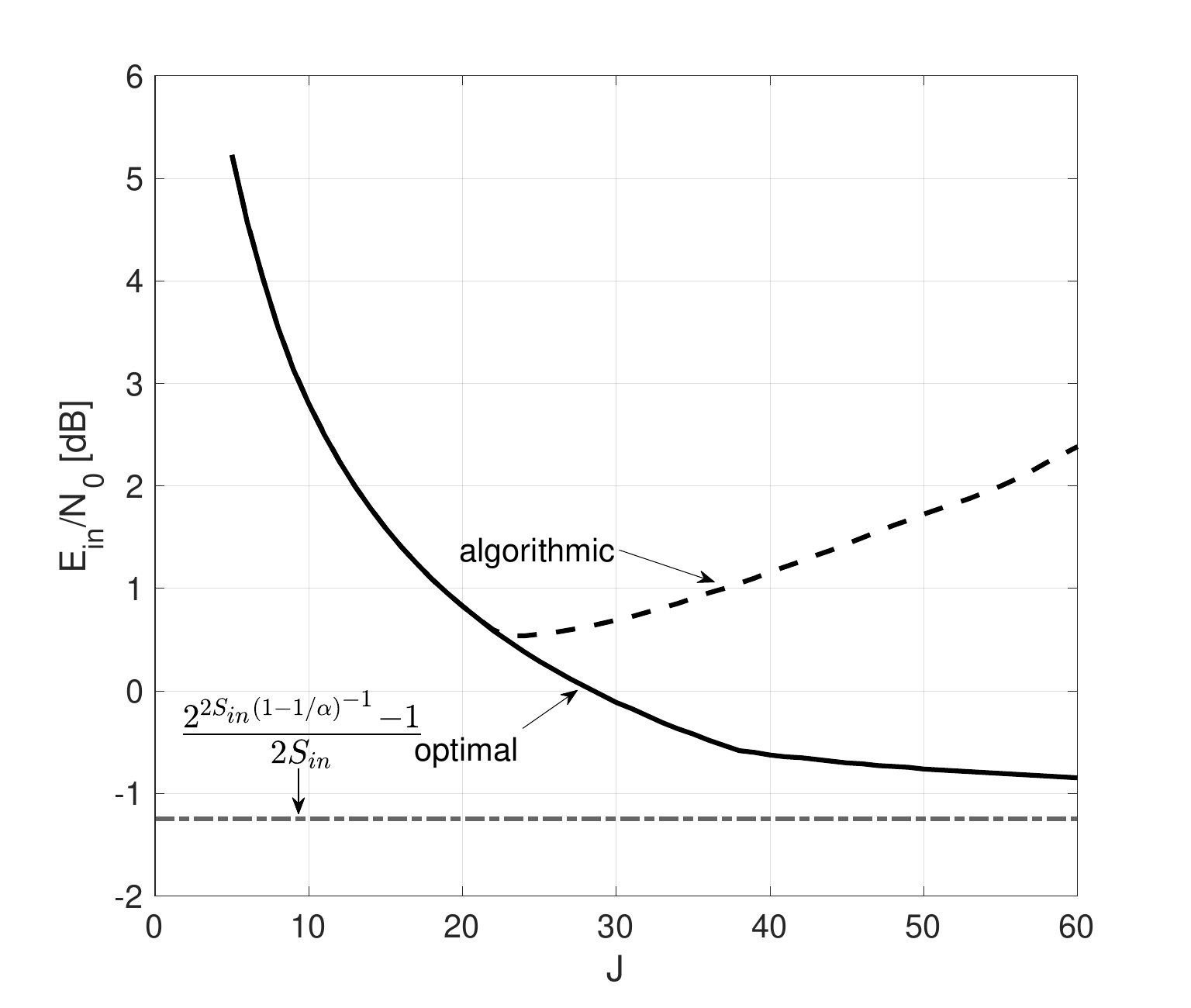}
    \caption{ Required $\mathcal{E}_\text{in}$ to reach specific target error probabilities
        with $\alpha = 2$ and $S_\text{in} = 2$
        according to Theorems \ref{thm:se_amp},\ref{thm:inner} and Claim \ref{conj:decoupling}.
        The results show that the AMP estimate becomes increasingly suboptimal once $J$
        passes a certain threshold around $J^*\approx 22$.
    }
    \label{fig:SE_J}
\end{figure}
\section{Optimizing the Power Allocation}
\label{sec:pa}
The foregoing asymptotic analysis has
important implications for the code design. We have empirically observed that
there is a critical number of users at which the required energy-per-bit increases sharply
and that this critical number gets smaller as $J$ grows larger.
According to the analytic insight developed in the previous section
this behavior is to be attributed to the sub-optimality of the decoder, see \figref{fig:SE_alg},
and we expect the required energy to decrease further with $J$ if an optimal decoder is used.
For single user sparse regression codes it is possible to get rid of the local minimum in the
potential function. This has been achieved through a non-uniform power allocation in
\cite{Rus2017} and through spatial coupling in \cite{Bar2016a, Bar2017a, Hsi2018a, Rus2020}.
More generally, it was shown
in \cite{Yed2012,Yed2014} that in a set of identical copies
of a recursive equation, which are described by some potential function,
when the copies of the recursion are coupled in a special way, the potential
of the coupled system 
has only one minimum, which coincides with the global minimum of the scalar potential
function, even if the uncoupled potential has a local, non-global minimum.
This phenomenon was termed \emph{threshold saturation} in the context
of belief-propagation decoding \cite{Kud2011}. 

We focus on the power allocation approach, because it is easy to obtain an optimized
power allocation for an AMP algorithm with separable denoising functions
and a given input distribution. We present a linear programming algorithm to optimize the
power allocation for the AMP algorithm that follows very closely the optimization procedure
of \cite{Cai2004a}, which was developed in the context of CDMA.

The denoising function in \eqref{eq:eta_add} was specific to a uniform power allocation $P_l = P/L$ for all
$l$. For a generic power allocation we replace the
componentwise denoising functions $f_{t,i}$ with
ones that depend on the section index:
\beq
f_{t,i}^l(x) = \frac{\sqrt{\hat{P}_l}}{Z(x)}\sum_{k=0}^{K_a}p_kk\exp\left(\frac{1}{2\tau_t^2}\left(x-k\sqrt{\hat{P_l}}\right)^2\right)
    \label{eq:eta_add_PA}
\eeq
where $\hat{P}_l = nP_l/L$ and
\beq
Z(x) = \sum_{k=0}^{K_a}p_k\exp\left(\frac{1}{2\tau_t^2}\left(x-k\sqrt{\hat{P_l}}\right)^2\right)
\eeq
To analyse the error probability of this modified AMP algorithm in the asymptotic limit
$L\to\infty$, we assume that the powers $P_l$ take values only in a finite set
$\{\Pi_1,...,\Pi_I\}$
and that the ratio of sections which use $\Pi_i$ is given by $\alpha_i = |\{l:P_l = \Pi_i\}|/L$
satisfying $\sum_{i=1}^I \alpha_i = 1$ and $\sum_{i=1}^I \alpha_i\Pi_i = P$.
We assume that these ratios stay constant as $L\to\infty$. According to the generalized
SE in \cite{Ber2020} the recursive equation which describes the behavior of the modified
AMP algorithm is given by
\beq
    \tau^2_{t+1} = \sigma_w^2 + \lim_{L\to\infty}\frac{\beta}{L2^J}
    \EE\left[\|\eta_t(\thetav + \tau_t\Zm) - \thetav\|_2^2\right]
    \label{eq:berthier}
\eeq
where $\sigma^2_w = P^{-1}$,
\beq
\thetav = (\thetav^1|...|\thetav^L) = \left(\left.\sqrt{\hat{P}_1}\sv^1\right|...\left|\sqrt{\hat{P}_L}\sv^L\right.\right)^\top
\eeq
is a rescaled version of $\sv$ and $\eta_t$ is componentwise given by $\eta_{t,i}^l = \sqrt{\hat{P}_l}f_{t,i}^l$. 
Since we choose $f_t^l = (f_{t,1}^l,...,f_{t,2^J}^l)$ to be separable the expected value in \eqref{eq:berthier} decouples as follows
\beq
\begin{split}
    &\EE\left[\|\eta_t(\thetav + \tau_tZ) - \thetav\|_2^2\right] = \\
    &\sum_{l=1}^L
    \sum_{j=1}^{2^J}
    \EE\left[\left( \sqrt{P}_lf_{t,j}^l\left(\sqrt{P}_ls^l_j + \tau_tZ\right) - \sqrt{P}_ls^l_j\right)^2\right] 
\end{split}
\eeq
As $L$ goes to infinity the sum over $l$ converges to its mean for each $j$
\beq
\begin{split}
    \lim_{L\to\infty}\frac{1}{L}
    &\sum_{l=1}^L
    \EE\left[\left( \sqrt{P}_lf_{t,j}^l\left(\sqrt{P}_ls^l_j + \tau_tZ\right) - \sqrt{P}_ls^l_j\right)^2\right] \\
    &= 
    \sum_{i=1}^I\alpha_i\hat{\Pi}_i\EE\left[\left(f_{t,i}^l\left(\sqrt{\Pi}_is + \tau_tZ\right) - s\right)^2\right] 
    \label{eq:pa_mse}
\end{split}
\eeq 
where now $s$ is a random variable distributed according to the marginal empirical distribution of $\sv$
and $\hat{\Pi}_i = n\Pi_i/L = J\Pi_i/R_\text{in}$.
This holds for each component $j$ and each $j$ has the same marginal distribution,
so the sum over $j$ becomes redundant.
\begin{remark}
    We can see that in this calculation the sums over $j$ and $l$ are interchangeable. If
    $2^J$ is large enough the sum over $j$ already converges to its mean value and the
    sum over $l$ becomes redundant. This explains heuristically why we can observe a good
    correspondence between the state evolution and the empirical performance even for
    small $L$. Technically the exponential scaling regime with $\beta\to\infty$ is not covered
    by the result of \cite{Ber2020}.
\end{remark}
Note, that the denoising functions $f_{t,i}^l$ were chosen precisely as the PME of $s$
in a scalar Gaussian channel like \eqref{eq:decoupled_scalar} with power $\sqrt{\hat{P}_l}$.
So they minimize the MSE in \eqref{eq:pa_mse} and by 
substituting $\tau^2 = \eta^{-1}$ we get the fixed point condition
\beq
\eta^{-1} = 1 + \beta\sum_{i=1}^I\alpha_i\hat{\Pi}_i\text{mmse}(\eta\hat{\Pi}_i).
\label{eq:SE_PA}
\eeq 
The function $\text{mmse}(t)$ is precisely
the same as in \eqref{eq:SE}. We can see that the right hand side of \eqref{eq:SE_PA}
is a linear combination of rescaled versions of the original MMSE function in
\eqref{eq:SE}. We can formulate the condition that \eqref{eq:SE_PA} has no local minima besides
the global minimum around $\eta = 1$ as follows:
\beq
\begin{split}
    \underset{\alphav}{\text{minimize}}\quad &\sum_{i=1}^I \alpha_i\Pi_i \\
    \text{subject to}\quad &1 + \beta\sum_{i=1}^I\alpha_i\hat{\Pi}_i\text{mmse}(\eta\hat{\Pi}_i) < \eta^{-1}-\epsilon\\
                           &\forall \eta \in [0,1-\delta],\\
                      &\sum_{i=1}^I \alpha_i = 1,\\
                      &\alpha_i \geq 0 
\end{split}
\label{eq:PA_opt}
\eeq
where $\epsilon,\delta>0$ are appropriately chosen slack variables. 
The optimization problem \eqref{eq:PA_opt} is a linear program
and therefore easily solvable. 
The discrete set of $\Pi_i$ is chosen as follows. For fixed $K_a,J,R_\text{in}$ we set
a target inner channel strength. We use Theorem \ref{thm:error_bound} to determine the smallest power
$P_\text{opt}$ such that for $\eta_\text{opt}$, the global minimizer of \eqref{eq:RS-scalar-potential},
it holds
that $\eta_\text{opt}\hat{P}_\text{opt}$ exceeds the target inner channel strength.
This $P_\text{opt}$ serves as lower bound on the set of $\Pi_i$. The upper bound
is chosen arbitrary, e.g. $5P_\text{opt}$. The $\Pi_i$ are then chosen as a
uniform discretisation of the interval $[P_\text{opt},5P_\text{opt}]$.
In \figref{fig:pa_optimization} we visualize this for $K_a = 300, J=20, R_\text{in}=0.0061$,
where we have chosen the slack parameters $\epsilon = 0.01, \delta = 0.1$. 
The solution of \eqref{eq:PA_opt} gives an optimal power distribution that
puts weight only two values, $\Pi_1 = P_\text{opt}$ and another value $P_*$ in a ratio
of $\alpha_1 \sim 0.81$ to $\alpha_2 \sim 0.19$. The total average power
$\alpha_1P_\text{opt} + \alpha_2P_*$ is about $0.5$dB smaller than $P_\text{alg}$,
the power at which \eqref{eq:RS-scalar-potential} has no local minimizers.
This means by letting one fifth of the
sections use a higher power it is possible to let the other four fifths of the section
use the optimal power without having a local convergence point.
In \figref{fig:pa_optimization} a) we plot
\beq
    g(\eta) = 1 + 
    \beta\left(\alpha_1\hat{P}_\text{opt}\text{mmse}(\eta\hat{P}_\text{opt})+\alpha_2\hat{P}_*\text{mmse}(\eta\hat{P}_*)\right)-\eta^{-1}
\eeq 
and its counterparts without power allocation. 
\figref{fig:pa_optimization} b) shows the integral of $g(\eta)$, which resemble the 
RS-potential \eqref{eq:RS-scalar-potential} with a non-uniform power allocation.
\begin{figure}
   \centering
   \subfloat[]{\includegraphics[width=\linewidth]{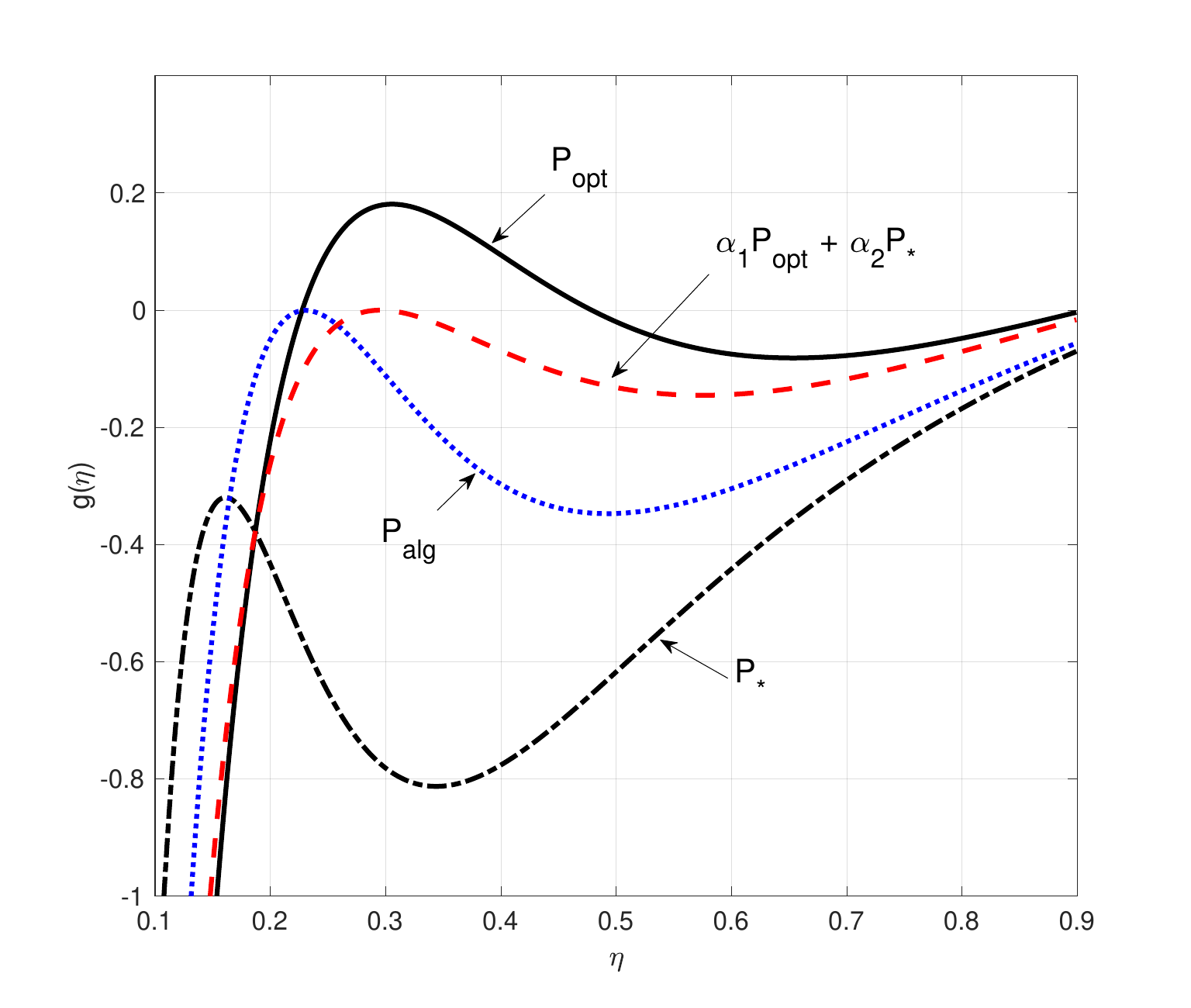}}\quad
   \subfloat[]{\includegraphics[width=\linewidth]{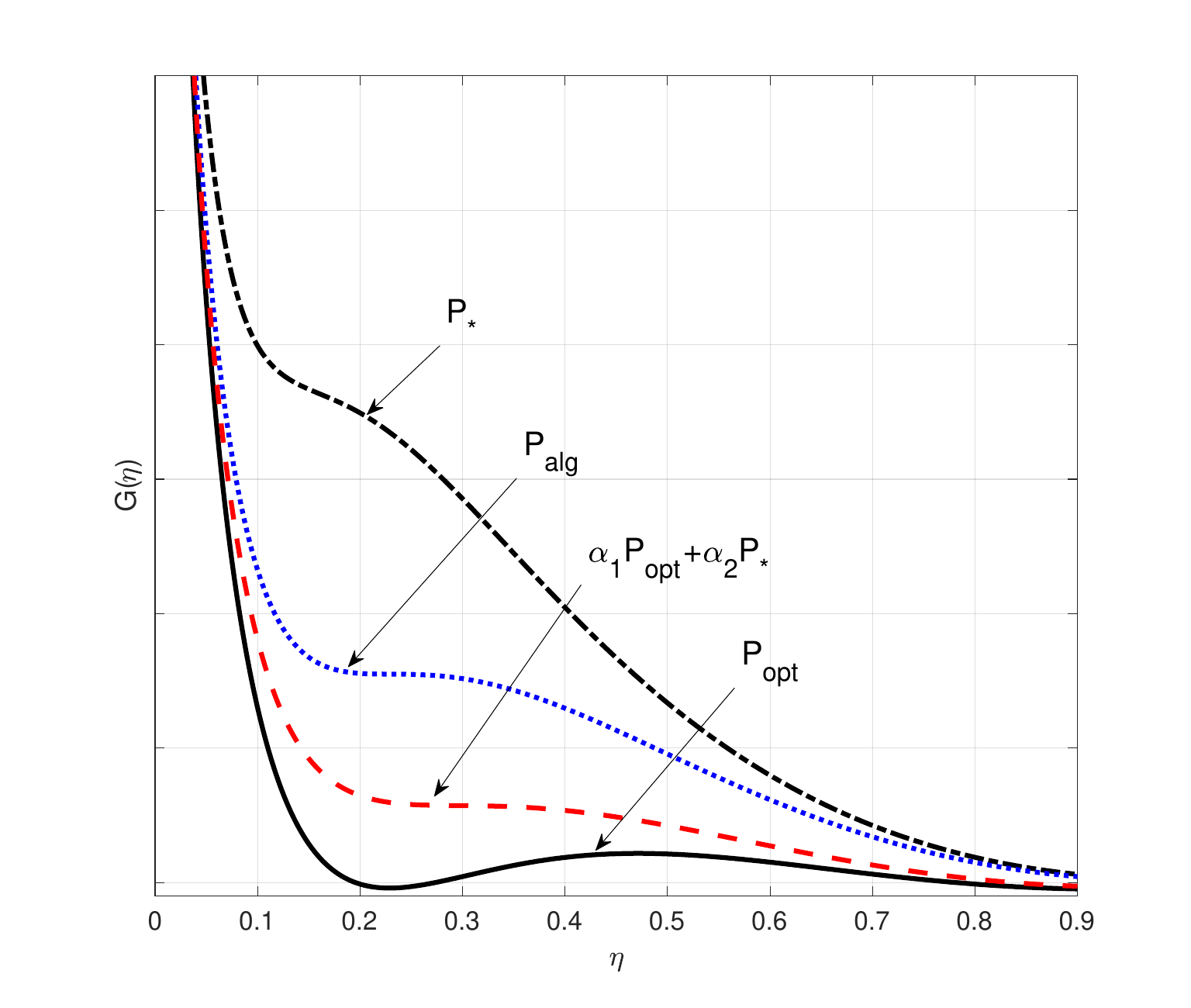}}
   \caption[Titel des Bildes]{Visualization of the solution of the optimization problem for
   $K_a = 300,R_\text{in} = 0.0061,J=20$. The solution puts weights only on two powers
   $P_\text{opt}$ and $P_*\sim 1.9P_\text{opt}$ with ratios $\alpha_1 = 0.81$ and $ \alpha_2 = 0.19$. The total
   power in the power allocated system is $\mathcal{E}_\text{in} = 1.6$dB while the algorithmic threshold
   is $\mathcal{E}_\text{in,alg} = 2.1$dB. So the power allocation in this case gives a gain of $0.5$dB.}
   \label{fig:pa_optimization}
\end{figure}
The efficiency of the power allocation is demonstrated in \figref{fig:p_e_2level} with
finite-length simulations. We choose $L=8, J=20$ and use the outer tree code with
0 parity bits in the first section, $20$ parity bits in the last, and alternating between 8 and 9 parity
bits in the remaining sections. This leaves a total of $89$ data bits. We choose those
numbers to stay below the typical number of 100 bits that is commonly used in IoT scenarios.
The blocklength is chosen as $n=26229$, which results in $R_\text{in} = 0.0061$ and
a per-user spectral efficiency of 
$R_\text{in}R_\text{out} = 0.0034$.
To distribute the power as closely as possible to the optimized power allocation obtained
above we choose two sections to have a power roughly twice as high as the remaining six
sections. We can see that by using a power allocation a gain of about 0.5-1 dB is achievable,
which matches the theoretical prediction.
With the same parameters as above but $K_a = 200$ we find that
$P_\text{alg} = P_\text{opt}$.
So the desired inner channel strength of $15$dB can be obtained by the AMP algorithm with 
a flat power allocation. In that case a non-uniform power allocation may be detrimental,
because it could introduce unwanted local minima into \eqref{eq:RS-scalar-potential}.
Indeed, simulations confirm that the 2-level power allocation that was effective for $K_a = 300$
actually worsens the performance for $K_a \leq 250$. This means that
the power allocation has to be tailored carefully to the expected parameters and
it only improves the performance if there is a gap between $P_\text{alg}$ and $P_\text{opt}$.  
\begin{figure}
    \centering
    \includegraphics[width=\linewidth]{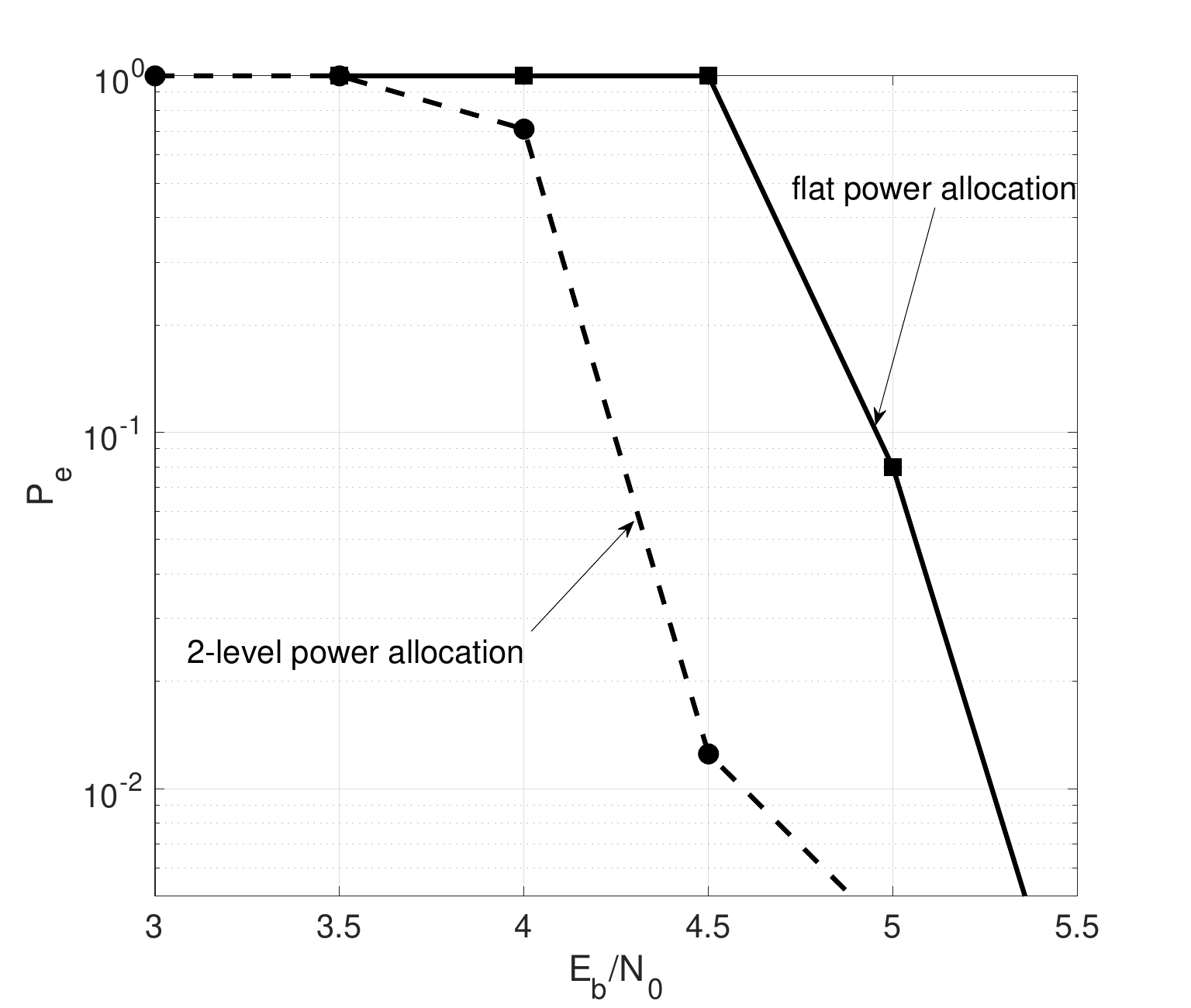}
    \caption{ $J=20, L=8, B=89$ bits, $n=26229, R_\text{out} = 0.55$, which leads to 
        $R_\text{in} = 0.0061$ and $\mu = 0.0034$. $P_e$ is the per-user-probability of error.
    }
    \label{fig:p_e_2level}
\end{figure}
\section{Considerations for Practical Implementation}
\label{sec:practical}
The denoising function in the AMP algorithm, \eqref{eq:eta_add} or its counterpart for non-uniform powers in \eqref{eq:eta_add_PA},
can be simplified if $K_a\ll2^J$.
Then the probabilities $p_k$ are very small for $k\geq 2$ and by neglecting them we
get the modified denoising function
\beq
f_{t,i}^\text{OR}(x) = \sqrt{\hat{P}}\left(1+\frac{p_0}{1-p_0}\exp\left(\frac{\hat{P}-2\sqrt{\hat{P}x}}{2\tau_t^2}\right)\right)^{-1}
\label{eq:or_estimator}
\eeq
which we denote with the suffix OR because it is the PME in the channel \eqref{eq:decoupled_or}.
In the parameter range of our simulations we can see almost no difference between the
full version \eqref{eq:eta_add} and the truncated OR-estimator \eqref{eq:or_estimator}.
For moderate values of $2^J$ compared to $K_a$,
we can improve on \eqref{eq:or_estimator} slightly by taking into account $p_2$:
\beq
\begin{split}
    f_{t,i}^\text{OR+}(x) &= \frac{\sqrt{\hat{P}}}{Z(x)}\left(p_1\exp\left(-\frac{\left(x-\sqrt{\hat{P}}\right)^2}{2\tau_t^2}\right)\right.\\
&\left.+2(1-p_0-p_1)\exp\left(-\frac{\left(x-2\sqrt{\hat{P}}\right)^2}{2\tau_t^2}\right)\right)
\label{eq:orplus_estimator}
\end{split}
\eeq
where 
\beq
\begin{split}
    Z(x) = 
    &p_0\exp\left(-\frac{x^2}{2\tau_t^2}\right)+p_1\exp\left(-\frac{\left(x-\sqrt{\hat{P}}\right)^2}{2\tau_t^2}\right) \\
    &+(1-p_0-p_1)\exp\left(-\frac{\left(x-2\sqrt{\hat{P}}\right)^2}{2\tau_t^2}\right)
\end{split}
\eeq
If necessary, further terms can be included in the same way.
\section{Finite-length Simulations}
\label{sec:sims}
In \figref{fig:ebn0}
the required $E_b/N_0$ to achieve a per-user probability $P_e<0.05$ is shown.
For the empirical curves 
with $J=15$ we use $B=100$ bits
and $n=30000$ real symbols. For $J=20$ we use $B = 89$ bits and
$n=26226$ real symbols. This results in a per-user spectral efficiency of
$\mu = 0.0033$ and $\mu = 0.0034$, which is the typical value used in comparable works
\cite{Pol2017,Pra2020a,Vem2017,Ama2020a,Mar2019}.
As an inner decoder we use the AMP algorithm \eqref{eq:amp} with \eqref{eq:orplus_estimator}
as denoiser.
After the inner decoder converged, in each section, the $K_a + \Delta$ largest entries
are declared as active and added to the list $\mathcal{S}_l$ for the outer tree code.
We use $\Delta = 50$.
For $J=15$ we use $L=16$ and the parity bits for the outer code are chosen as:
$\piv = (0,7,8,8,9,...,9,13,14)$. For $J=20$ we use $L=8$ and a parity bit distribution
$\piv = (0,9,8,9,8,9,8,20)$.
The outer rates are $R_\text{out} = 0.4167$ and $R_\text{out} = 0.5563$ respectively.
For the SE curves we first estimate the required effective inner channel strength
by setting the error rates in the outer channel to $p_\text{md} = P_e/L$ and
$p_\text{fa} = \Delta/2^J$. The required inner channel strength is then calculated by
\eqref{eq:tradeoff}.
We use the potential \eqref{eq:RS-scalar-potential} to estimate
the power to achieve the required inner channel strength. 
For the curves with power allocation we use the
method of Section \ref{sec:pa} to find the optimal power allocation for $K_a = 300$.
We can see
that in both cases, $J=15$ and $J=20$, the required power decreases for $K_a = 300$
but increases for all other values of $K_a$. So the power allocation has to be adapted to the
expected number of users. 
The empirical values match the theoretically estimated SE curves very well, which confirms
the precision of the asymptotic analysis, even though the number of sections is very small. 

The obtained value of $4.3$ dB for $K_a = 300$ is at the point of writing $0.7$ dB better
than the best reported value of $5$ dB, which was achieved in \cite{Ama2020c}.
For smaller values of $K_a$ other coding schemes have achieved better results, the best of
which at the point of writing are \cite{Pra2020a} and \cite{Mar2019},
but both of those schemes have shown a rapid increase
in required energy as $K_a$ grows large.
In \cite{Ama2020c} an enhanced version of
the discussed concatenated coding scheme was presented, where another outer code was used that
enabled the passing of soft decoding information between the AMP decoder and the outer decoder,
resulting in an turbo-like iterative decoding scheme, alternating between inner and outer
decoder. With this type of decoding the required power for
$K_a \leq 250$ is reduced significantly,
but for $K_a = 300$ the required power is still around $5$dB.
\begin{figure}
    \centering
    \includegraphics[width=\linewidth]{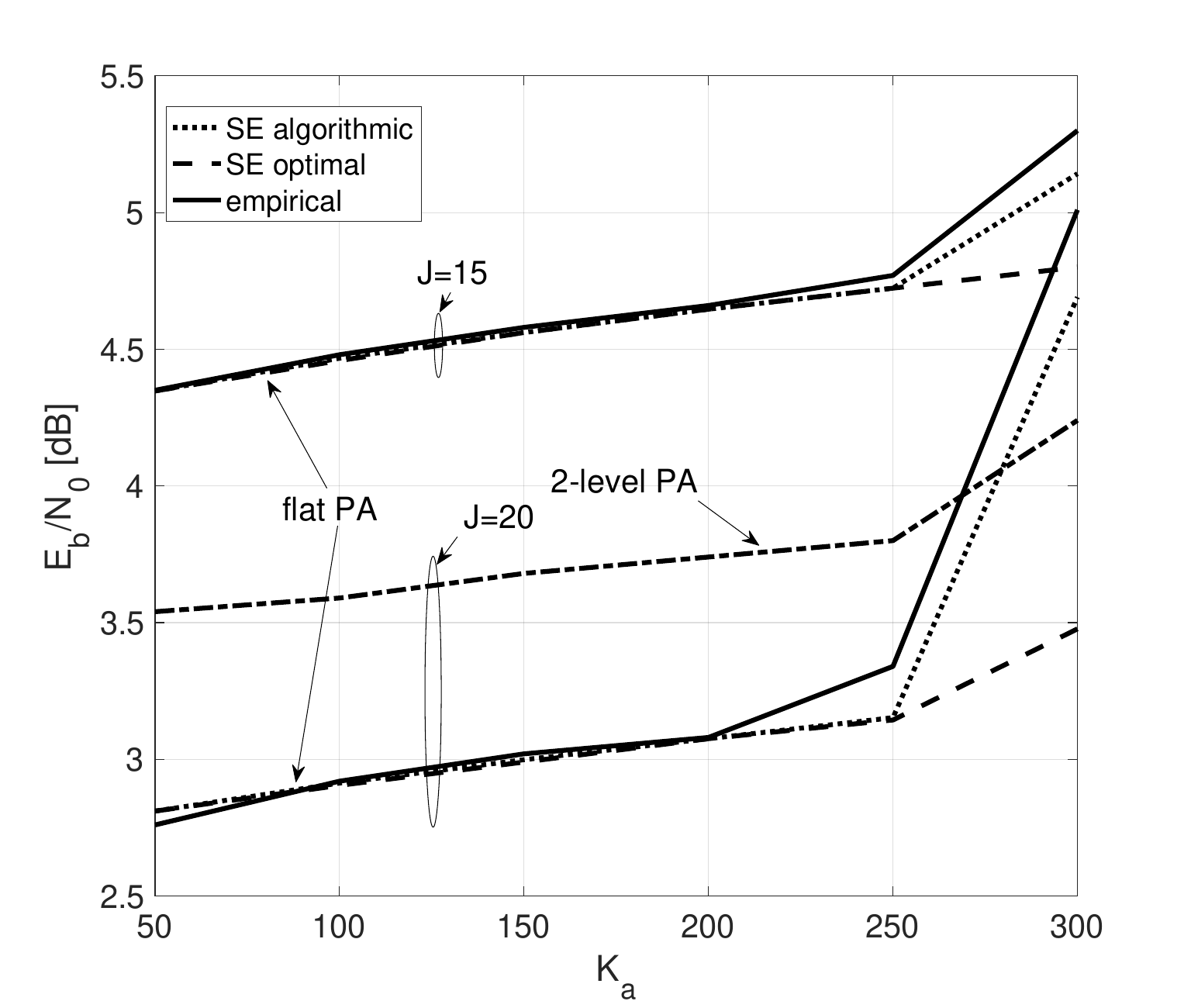}
    \caption{ Required $E_b/N_0$ to achieve a per-user error probability $P_e<0.05$.
        For $J=15$ we used $B=100$, $n=30000$,
        and for $J=20$ we used $B=89$, $n=26229$.
    }
    \label{fig:ebn0}
\end{figure}

\section{Summary and Outlook}
In this work we have introduced a concatenated coding construction that
extends the concept of sparse regression codes 
to the unsourced random access scenario. In this construction an inner code is used
as an efficient single user channel code for the AWGN channel and an outer code is
used to resolve the multiple access interference. The structural similarity to the coupled
compressed sensing scheme allowed us to use the tree code presented in
\cite{Ama2020a} as an outer code. We use the AMP algorithm as inner decoder, for which
we have introduced a low-complexity approximation to the Bayesian optimal
denoiser. We introduced a decomposition of the channel into an inner and an outer channel
to analyse the asymptotic limits under optimal decoding. Furthermore, we calculated
the
asymptotically required energy-per-bit of the inner AMP decoder and compared
it to the optimal decoder.
Finite-length simulations show that the calculated results describe the actual
required energy-per-bit for a fixed per-user error probability very precisely.
We find that as $J\to\infty$, where $2^J$ is the size of the outer alphabet,
the achievable sum-rates of the concatenated code converge to the Shannon limit, even
if the number of active users $K_a$ grows to infinity simultaneously, but much faster than
$J$. This is in stark contrast to the typical information theoretic limit, where
the size of the message is considered to be much larger than the number of active users.
Therefore, it is noteworthy that even under short message length (compared to the number of
users) and no coordination between users the Shannon limit can be achieved. 

Unfortunately, also the difference in required energy between the AMP decoder and the
optimal decoder grows rapidly with $J$ once $J$ surpasses a certain value that depends
on the rate. When there is a difference in performance between the AMP and the optimal inner
decoder, a non-uniform power allocation can be used to improve the performance of the AMP
decoder. We present a linear programming algorithm to find an optimized power allocation.
Although for small sum-spectral efficiencies existing U-RA coding schemes like \cite{Pra2020a,Mar2019}
are more energy-efficient than the presented scheme,
we show that for a sum-spectral efficiency of 1 bit/c.u. and $K_a = 300$ users the presented
approach improves on existing ones by almost 1 dB.
The good performance at high spectral efficiencies and the availability of a precise
analysis make the presented coding scheme stand out among the existing U-RA approaches
and therefore an interesting candidate for massive MTC.

The extension of the presented coding scheme and analysis to more general channel models
incorporating fading, asynchronicity or multiple receivers seems to be in reach and promising
research directions. Furthermore, the presented analysis can be seen as a basis to
analyse the more complex turbo-like decoder of \cite{Ama2020c}.

\appendices

\section{Optimal product distribution}
\label{appendix:product}
\begin{lemma}
Let $\Omega \subset \RR$ be some discrete set.
Let $p(\sv)$ with $p: \Omega^{2^J} \to \RR_+$ be a probability mass function on $\Omega^{2^J}$.
Let $p_{s_1}(s_1),....,p_{s_{2^J}}(s_{2^J})$ denote the marginals of $p(\sv)$.
Further let,
\beq
\mathcal{P}_\text{prod} :=
\left\{q(\sv) = \prod_{i=1}^{2^J}q_i(s_i), q_i:\Omega\to\RR_+\middle|\sum_{s\in\Omega} q_i(s) = 1\right\}
\eeq
denote the space of product distributions
on $\Omega^{2^J}$. Then
\beq
\argmin_{q \in \mathcal{P}_\text{prod}}D(p\ \|\ q)
= \prod_{i=1}^{2^J}p_{s_i}(s_i)
\eeq
\end{lemma}
\begin{proof}
    For a product distribution $q \in \mathcal{P}_\text{prod}$,
    $D(p\ \|\ q)$ can be expressed as:
\begin{align}
D(p\ \|\ q)
&= \sum_\sv p(\sv) \log\frac{p(\sv)}{q(\sv)} \\
&= \sum_\sv p(\sv) \log\frac{p(\sv)}{\prod_ip_{s_i}(s_i)}\frac{\prod_ip_{s_i}(s_i)}{q(\sv)} \\
&= D\left(p\ \middle\|\ \prod_ip_{s_i}\right) + \sum_\sv\sum_{i=1}^{2^J}p(\sv)\log\frac{p_{s_i}(s_i)}{q_i(s_i)} 
\end{align}
The first term is independent of $q$ and the second term can be rewritten as
\begin{align}
    &\sum_\sv\sum_{i=1}^{2^J}p(\sv)\log\frac{p_{s_i}(s_i)}{q_i(s_i)}  \\
    &= \sum_{i=1}^{2^J}\sum_{s_i}\left(\sum_{\sv\setminus s_i}p(\sv)\right)\log\frac{p_{s_i}(s_i)}{q_i(s_i)}\\
    &= \sum_{i=1}^{2^J}\sum_{s_i}p_{s_i}(s_i)\log\frac{p_{s_i}(s_i)}{q_i(s_i)}\\
&= \sum_{i=1}^{2^J}D(p_{s_i}\ \|\ q_i)
\end{align}
which is non-negative and minimized by $q_i \equiv p_{s_i}$.
\end{proof}

\section{$D_\text{KL}$(Mult $\|$ Bin)}
\label{appendix:kl}
A vector $(z_1,...,z_{2^J})$ is called multinomial distributed with parameter $n$
and probabilities
$p_1,...,p_{2^J}$ if
\begin{align}
    p(\zv) &= \PP(Z_1 = z_1,..,Z_{2^J} = z_{2^J}) \\
         &=\begin{cases}
             \frac{n!}{z_1!\cdots z_{2^J}!} p_1^{z_1}\cdots p_{2^J}^{z_{2^J}} & \quad \sum_{i=1}^{2^J} z_i = n \\
        0 & \quad \text{else}
    \end{cases}.
    \label{eq:multinomial_app}
\end{align}
It follows from the multinomial theorem, that the distribution is normalized
$\sum_\zv p(\zv) = 1$.
An important property of the multinomial distribution is that the marginals follow a binomial
distribution:
\beq
p(Z_i = z_i) = \sum_{\zv\backslash z_i}\PP(z_1,...,z_{2^J}) = \binom{n}{z_i} p_i^{z_i}(1-p_i)^{n-z_i}
    \label{eq:marginal}
\eeq
with covariance given by $\text{cov}(Z_i,Z_j) = -np_ip_j$.

Let $n=K_a$ and
$p_i = 2^{-J}$ for all $i = 1,...,K_a$ and let $q(Z)$
denote the binomial distribution with parameters $n=K_a$ and $p_i =  2^{-J}$.
Then the marginals of $p$ are all identical and equal to $q(Z)$ and it holds 
\begin{align}
    D_\text{KL}\left(p \ \middle\|\ \prod_{i=1}^{2^J} q_i \right) 
    &= \sum_\zv p(\zv)\log_2 \frac{p(\zv)}{q_i(z_i)} \\
    &= - H(p) + 2^JH(q)
    \label{eq:kl_1}
\end{align}
where $H(p)$ denotes the entropy of the multinomial distribution $p$ and $H(q)$
the entropy of the binomial distribution $q$.
Both entropies are well known and given by
\begin{align}
    &H(p) = JK_a - 
    \log_2 K_a!\\
    &+ 2^J\sum_{t=0}^{K_a}{{K_a}\choose{t}}\left(2^{-J}\right)^t\left(1-2^{-J}\right)^{K_a-t}\log_2 t! 
    \label{eq:h_add}
\end{align}
and
\begin{align}
    H(q) 
    &= -\log_2 K_a! + \EE_q [\log Z!] + \EE_q [\log (K_a-Z)!]\\
    &+ J \EE_q [Z] - \EE_q [K_a-Z] \log_2 (1-2^{-J}) 
    \label{eq:h_binom}
\end{align}
We have $\EE_q [Z] = K_a/2^J$ and $\EE_q [K_a-Z] = K_a - K_a/2^J$.
In the limit for large $J$, we can expand $\EE_q [\log_2 (K_a-Z)!]$ in terms of $2^{-J}$ and get:
\begin{align}
    \EE_q [\log_2 (K_a - Z)!] = \log_2 K_a! - \frac{K_a}{2^J}\log_2 K_a + \mathcal{O}\left(\frac{1}{2^{2J}}\right).
\end{align}
Inserting this, \eqref{eq:h_add} and \eqref{eq:h_binom} into \eqref{eq:kl_1},
many terms cancel and we get:
\begin{align}
     &D_\text{KL}\left(p \ \middle\|\ \prod_{i=1}^{K_a} q_i\right) \\
     &= \log_2 K_a! - K_a(2^J-1)\log_2 (1 - 2^{-J}) - K_a\log_2 K_a
\end{align}
Using $\log_2 (1 - 2^{-J}) = -\log_2(e)/2^J + \mathcal{O}(1/2^{2J})$ for large $J$
and the Stirling approximation $\log_2 K_a! = K_a\log_2 K_a - K_a\log_2 e + \mathcal{O}(\log K_a)$ we get
\beq
D_\text{KL}\left(p(\xv)\ \middle\|\ \prod_{i=1}^{K_a} q_i(x_i)\right) 
  = \mathcal{O}(\log K_a) - \frac{K_a\log_2 e}{2^J}.
\eeq
which implies \eqref{eq:kl}.

\section{Exponential $L^1$ convergence implies exponential convergence in measure}
\label{appendix:eta_bound}
\begin{theorem}
    Let $(f_J)_{J=1,2,...}$ be a sequence of integrable functions s.t.
    \beq
    \|f_J - f\|_{L^1} \leq \frac{c}{2^J}
    \label{eq:exp_L1_condition}
    \eeq 
    for some constant $c>0$ and all large enough $J$. For any $\delta > 0$
    there is a $J_\delta$ such that for all $J\geq J_\delta$
    \beq
    |f_J(t) - f(t)| \leq \frac{\sqrt{J}}{\delta2^{J}}
    \eeq 
    holds for all $t$ except for a set of size $\mathcal{O}(\delta J^{-1/2})$.
    $\hfill\square$
\end{theorem}
\begin{proof}
    Let $\epsilon > 0$ then
    \beq
    \sum_{j=J}^\infty \|f_j - f\|_{L^1} \leq \frac{c}{2^{J-1}} < \epsilon\frac{\sqrt{J}}{\delta2^{J}}
    \eeq
    holds for all 
    \beq
    J > J_\delta = \left(\frac{2c\delta}{\epsilon}\right)^2
    \label{eq:J_delta}
    \eeq
    where we have used condition \eqref{eq:exp_L1_condition} and the formula
    \beq
    \sum_{j=J}^\infty 2^{-j} = 2^{-(J-1)}.
    \eeq 
    Now define the sets 
    \beq
    A_j = \left\{t: |f_j(t)- f(t)|>\frac{\sqrt{j}}{\delta2^{j}}\right\}
    \eeq 
    and let $\mu(A_j)$ denote the Lebesgue measure of $A_j$. Then it follows
    from elementary properties of the integral that
    \beq
    \sum_{j=J}^\infty \mu(A_j)\frac{\sqrt{J}}{\delta2^{J}} < 
    \sum_{j=J}^\infty \|f_j - f\|_{L^1} < \epsilon\frac{\sqrt{J}}{\delta2^{J}}
    \eeq
    and so
    \beq
        \mu\left(\bigcup_{j=J}^\infty A_j\right) \leq
        \sum_{j=J}^\infty \mu(A_j) < \epsilon
        \label{eq:mu_bound}
    \eeq
    Let $\mathcal{A}_J:=\bigcup_{j=J}^\infty A_j$,
    then
    \beq
        \mathcal{A}_J^c = 
        \left\{t: |f_j(t)- f(t)|\leq\frac{\sqrt{J}}{\delta2^{J}},\ \forall j\geq J\right\}
    \eeq 
    and so \eqref{eq:mu_bound} states that $\mu(\mathcal{A}_J)$, the measure of the set of points on which
    the pointwise convergence does not hold, can be made arbitrarily small. More precisely,
    it follows from \eqref{eq:J_delta} that
    \beq
    \mu(\mathcal{A}_J) = \mathcal{O}\left(\frac{\delta}{\sqrt{J}}\right)
    \eeq   
\end{proof}
\section{Proof of Theorem \ref{thm:or_mmse}}
\label{appendix:or_mmse}
Let 
\beq
q(k) = p_k
\eeq
for $k = 0,..,K_a$, where $p_k$
are the binomial probabilities defined in \eqref{eq:binomial} and
\begin{align}
    q_\text{OR}(0) &= p_0 \\
    q_\text{OR}(1) &= 1 - p_0.
\end{align}
Let $r,s,z$ be jointly distributed according to the Gaussian model 
\beq
    r = \sqrt{t}s + z,
\eeq
with $z\sim \mathcal{N}(0,1)$ independent of $s$ for some fixed $t\geq0$ 
and $s$ distributed according to $q$.
Let $\text{mmse}(t)$ be the MMSE of estimating $s$ from the Gaussian observation $r$
and let $f(r)$ be the PME of $s$, given by
\beq
f(r) = \frac{1}{Z(r)}\sum_{k=0}^{K_a}p_kke^{-\left(r-k\sqrt{t}\right)^2/2}
\eeq
with
\beq
Z(r) = \sum_{k=0}^{K_a}p_ke^{-\left(r-k\sqrt{t}\right)^2/2}
\eeq
Let
$f^\text{OR}(r)$ be the mismatched PME, which estimates $s$ from $r$ assuming that $s$ is
distributed according to $q_\text{OR}$. It is given by
\beq
f^\text{OR}(r) 
= \frac{(1-p_0)e^{-\left(r-\sqrt{t}\right)^2/2}}{Z^\text{OR}(r)}
\label{eq:for}
\eeq
with
\beq
Z^\text{OR}(r) = p_0e^{-r^2/2} + (1-p_0)e^{-\left(r-\sqrt{t}\right)^2/2}
\eeq
Let $\text{mmse}_\text{OR}(t)$ be the 
mean square error of $f^\text{OR}(r)$. Since $f^\text{OR}(r)$ is mismatched we have 
$\text{mmse}(t) \leq \text{mmse}_\text{OR}(t)$
and it holds:
\begin{align}
    0 &\leq \text{mmse}_\text{OR}(t) - \text{mmse}(t) \\
      &= \sum_{k=0}^{K_a}p_k \EE\left\{\left[k - f^\text{OR}\left(\sqrt{t}k + z\right)\right]^2 
      - [k - f(\sqrt{t}k + z)]^2\right\} \\
      &= p_0 \EE\left\{\left[f^\text{OR}(z)\right]^2 - \left[f(z)\right]^2\right\}
        \label{eq:mse_1}\\
        &+ p_1\EE\left\{\left[1 - f^\text{OR}\left(\sqrt{t} + z\right)\right]^2 
        - \left[1 - f\left(\sqrt{t} + z\right)\right]^2\right\}
        \label{eq:mse_2}\\
      &+ \sum_{k=2}^{K_a}p_k \EE\left\{\left[k - f^\text{OR}\left(\sqrt{t}k + z\right)\right]^2 
        - \left[k - f\left(\sqrt{t}k + z\right)\right]^2\right\} \label{eq:mse_3}
\end{align}
We can bound the terms \eqref{eq:mse_1} - \eqref{eq:mse_3} individually.
Since $0\leq f^\text{OR}(r) \leq 1$, \eqref{eq:mse_3} is bound by
\begin{align}
    \eqref{eq:mse_3} &\leq \sum_{k=2}^{K_a}p_k k^2 \\
                     &= \text{Var}(s) + [\EE(s)]^2 - p_1 \\
                     &= \frac{K_a}{2^J}\left(1 - \frac{K_a}{2^J}\right) + \frac{K_a^2}{2^{2J}}
                     - \frac{K_a}{2^J}\left(1-2^{-J}\right)^{K_a-1} \\
                     &= \mathcal{O}\left(\frac{K_a^2}{2^{2J}}\right)
\end{align}
For the remaining terms we split the expected values over $z$ in \eqref{eq:mse_1} and $\eqref{eq:mse_2}$ in two parts,
depending on whether $Z(r) = Z(\sqrt{t} + z)$ is smaller or larger than $Z^\text{OR}(r)$.
\subsubsection{$Z(r) \leq Z^\text{OR}(r)$} We have:
\begin{align}
    f^\text{OR}(r) &= \frac{1}{Z^\text{OR}(r)}(1-p_0)e^{-\left(r - \sqrt{t}\right)^2/2} \\
                         &\leq \frac{1}{Z(r)}(1-p_0)e^{-\left(r - \sqrt{t}\right)^2/2}
\end{align}
and
\begin{align}
    f(r) &=
    \frac{1}{Z(r)}\sum_{k=0}^{K_a}p_kke^{-\left(r - \sqrt{t}k\right)^2/2} \\
                         &\geq \frac{1}{Z(r)}p_1e^{-\left(r - \sqrt{t}\right)^2/2}
\end{align}
because all the summands are non-negative. It follows for all $r$:
\begin{align}
    &[f^\text{OR}(r)]^2 - [f(r)]^2  \\
    &\leq \left[\frac{p_1e^{-\left(r-\sqrt{t}\right)^2/2}}{Z(r)}\right]^2\left[\left(\frac{1-p_0}{p_1}\right)^2 - 1\right] \\
    &\leq \left[\left(\frac{1-p_0}{p_1}\right)^2 - 1\right] \\
    &= \left[\left(\frac{1 - (1 - 2^{-J})^{K_a}}{\frac{K_a}{2^J}(1-2^{-J})^{K_a-1}}\right)^2 - 1\right] \\
    &\leq \mathcal{O}\left(\frac{K_a^2}{2^{2J}}\right)
    \label{eq:mse_bound_1}
\end{align}
which bounds the integral in  $\eqref{eq:mse_1}$ on the set $\{Z(z)\leq Z^\text{OR}(z)\}$.
For \eqref{eq:mse_2} notice that 
\begin{align}
    &\left[1 - f^\text{OR}(r)\right]^2  - [1 - f(r)]^2  \\
    &= f^\text{OR}(r)^2 - f(r)^2 
    + 2\left[f(r) - f^\text{OR}(r)\right]
    \label{eq:mse_square}
\end{align}
The first term was already bound in \eqref{eq:mse_bound_1},  we bound the second term by
\begin{align}
    &f(r) - f^\text{OR}(r)  \\
    &= \left[1 - f^\text{OR}(r)\right] - [1 - f(r)]  \\
    &= \frac{p_0e^{-r^2/2}}{Z^\text{OR}(r)} 
    - \frac{p_0e^{-r^2/2} +
    \sum_{k=1}^{K_a}p_k(1-k)e^{-\left(r-\sqrt{t}k\right)^2/2}}{Z(r)} \\
    &\leq 
    \frac{1}{Z(r)}\sum_{k=2}^{K_a}p_k(k-1)e^{-\left(r-\sqrt{t}k\right)^2/2}  \\
    &\leq 
    \frac{1}{Z(r)}\max_{k\geq2}\left\{e^{-\left(r-\sqrt{t}k\right)^2/2}\right\}
    \sum_{k=2}^{K_a}kp_k \\
    &=\frac{1}{Z(r)}\max_{k\geq2}\left\{e^{-\left(r-\sqrt{t}k\right)^2/2}\right\}
    \left(\frac{K_a}{2^J} - p_1\right) \\
    &=\frac{e^{-\left(r-\sqrt{t}k^*\right)^2/2}}{p_1e^{-\left(r-\sqrt{t}\right)^2/2}}
    \mathcal{O}\left(\frac{K_a^2}{2^{2J}}\right) 
    \label{eq:mse_bound_2}
\end{align}
where $k^* = \argmax_{k\geq2}\left\{\exp(-(r-\sqrt{t}k)^2/2)\right\}$.
For the last line, notice that $Z(r) \geq p_k\exp(-(r-\sqrt{t}k)^2/2)$ for all $k$ and
especially for $k=1$.
In \eqref{eq:mse_bound_2} we have $r = \sqrt{t} + z$. For the expected value in \eqref{eq:mse_2} we
have to integrate over $z$, restricted to the values of $z$ for which $Z(r)\leq Z^\text{OR}(r)$.
Since 
\begin{align}
    \int_{-\infty}^\infty e^{-z^2/2}\frac{e^{-(z-\sqrt{t}(k^*-1))^2/2}}{p_1e^{-z^2/2}} = \frac{1}{p_1}
\end{align}
and the integrand is non-negative, the same integral, restricted to $\{z:Z(\sqrt{t}+z) \leq Z^\text{OR}(\sqrt{t}+z)\}$,
is also bounded by $1/p_1$.
\subsubsection{$Z(r) > Z^\text{OR}(r)$}
Let $r$ be such that, $Z(r) > Z^\text{OR}(r)$. It holds that
\begin{align}
        f^\text{OR}(r) 
        &= 1 - \frac{p_0e^{-r^2/2}}{Z^\text{OR}(r)} \\
        &\leq 1 - \frac{p_0e^{-r^2/2}}{Z(r)} \\
        &= \frac{1}{Z(r)}
        \sum_{k=1}^{K_a}p_ke^{-\left(r-k\sqrt{t}\right)^2/2} \\
        &\leq \frac{1}{Z(r)}
        \sum_{k=1}^{K_a}kp_ke^{-\left(r-k\sqrt{t}\right)^2/2} \\
        &= f(r) 
\end{align}
Since both terms are non-negative we get
\beq
    f^\text{OR}(r)^2 - f(r)^2 \leq 0
    \label{eq:mse_bound_3}
\eeq
which, together with \eqref{eq:mse_bound_1}, shows that \eqref{eq:mse_1} is bounded
by a term of order $\mathcal{O}\left(\frac{K_a^2}{2^{2J}}\right)$.

For \eqref{eq:mse_2}, the same argumentation as in 
\eqref{eq:mse_square} holds and it remains to bound
$f(r) - f^\text{OR}(r)$. We have that
\begin{align}
    &f(r)
    - f^\text{OR}(r) \\
    &= \frac{1}{Z(r)}
    \sum_{k=1}^{K_a}kp_ke^{-\left(r-k\sqrt{t}\right)^2/2} 
    - \frac{1}{Z^\text{OR}(r)}
    (1-p_0)e^{-\left(r-\sqrt{t}\right)^2/2} \\
    &\leq \frac{1}{Z(r)}
    \sum_{k=1}^{K_a}kp_ke^{-\left(r-k\sqrt{t}\right)^2/2} 
    - \frac{1}{Z(r)}
    (1-p_0)e^{-\left(r-\sqrt{t}\right)^2/2} \\
    &\leq \frac{1}{Z(r)}
    \sum_{k=2}^{K_a}kp_ke^{-\left(r-k\sqrt{t}\right)^2/2} \\
    &\leq\frac{1}{Z(r)}
    \max_{k\geq2}\left\{e^{-\left(r-\sqrt{t}k\right)^2/2}\right\}\left(\frac{K_a}{2^J} - p_1\right) \\
    &=\frac{e^{-\left(r-\sqrt{t}k^*\right)^2/2}}{p_1e^{-\left(r-\sqrt{t}\right)^2/2}}
    \mathcal{O}\left(\frac{K_a^2}{2^{2J}}\right) 
    \label{eq:mse_bound_b}
\end{align}
This is the same term as in \eqref{eq:mse_bound_b}, for which we have
shown that its expected value over $z$ is bounded by $p_1^{-1}\mathcal{O}(K_a^2/2^{2J})$.
\eqref{eq:mse_bound_b}, together with \eqref{eq:mse_bound_2}, shows that $f(r)-f^\text{OR}(r)$ is
bounded by  $p_1^{-1}\mathcal{O}(K_a^2/2^{2J})$ for all $r$.
\eqref{eq:mse_bound_2} and \eqref{eq:mse_bound_3}
show that $f^\text{OR}(r)^2 - f(r)^2 = \mathcal{O}(K_a^2/2^{2J})$ for all $r$, and
therefore, by \eqref{eq:mse_square}, also \eqref{eq:mse_2} is bound
by a term of order $\mathcal{O}\left(\frac{K_a^2}{2^{2J}}\right)$.
This concludes the proof of Theorem \ref{thm:or_mmse}.
\section{Proof of Theorem \ref{thm:limit}}
\label{appendix:proof_limit}
The RS-potential \eqref{eq:RS-OR-potential}, rescaled by $\beta/2^J$
takes the form
\beq
i^\text{RS}(\eta) = \frac{R_\text{in}2^J}{J}I(\eta\hat{P}) +
\frac{\log_2 e}{2}[(\eta - 1) - \ln \eta]
\label{eq:rs_potential_rescaled}
\eeq
with the mutual information
\beq
I(\eta\hat{P}) := I(X;Y) = H(Y)-H(Y|X)
\eeq
for
$P(X=0) = p_0$, $P(X=1) = 1-p_0$ and $Y = (\eta\hat{P})^\frac{1}{2}X + Z$,
for $Z\sim\mathcal{N}(0,1)$ independent of $X$. 
The mutual information $I(\eta\hat{P})$ can be evaluated as follows.
First, note that in an additive channel $H(Y|X) = H(Z)$,
so $H(Y|X)$ is independent of $\eta$ and therefore we can ignore it when evaluating $i^\text{RS}(\eta)$.
The distribution of $Y$ is given by
\beq
\begin{split}
    p(y) &= p_0p(y|x=0) + (1-p_0)p(y|x=1)\\ 
         &= \frac{p_0}{\sqrt{2\pi}}\exp\left(-\frac{y^2}{2}\right)
         + \frac{1-p_0}{\sqrt{2\pi}}\exp\left(-\frac{1}{2}\left(y-(\eta\hat{P})^\frac{1}{2}\right)^2\right),
\end{split}
\eeq
so the differential output entropy $H(Y) = -\int p(y)\log_2 p(y)\mathrm{d}y$ can be split into the
sum of two parts. Define $H_0$ and $H_1$ respectively by
\beq
H_0 := -\frac{1}{\sqrt{2\pi}}\int_{-\infty}^\infty \exp\left(-\frac{y^2}{2}\right)\log_2(p(y))\mathrm{d}y
\label{def:H0}
\eeq
and
\beq
\begin{split}
    H_1 &:= -\frac{1}{\sqrt{2\pi}}\int_{-\infty}^\infty
    \exp\left(-\frac{1}{2}\left(y-(\eta\hat{P})^\frac{1}{2}\right)^2\right)\log_2(p(y))\mathrm{d}y\\
    &= -\frac{1}{\sqrt{2\pi}}\int_{-\infty}^\infty
    \exp\left(-\frac{y^2}{2}\right)\log_2\left(p\left(y+(\eta\hat{P})^\frac{1}{2}\right)\right)\mathrm{d}y
\end{split}
\label{def:H1}
\eeq
such that the following relation holds:
\beq
I(\eta\hat{P}) = p_0H_0 + (1-p_0)H_1.
\eeq
Taking into account the scaling factor in \eqref{eq:rs_potential_rescaled}
and using that $\lim_{J\to\infty}2^J(1-p_0) = K_a$ and $\lim_{J\to\infty}p_0 = 1$
we get that
\beq
\begin{split}
    \lim_{J\to\infty}\frac{R_\text{in}2^J}{J}I(\eta\hat{P}) 
    &= \lim_{J\to\infty}\left(\frac{R_\text{in}2^J}{J}H_0 + \frac{S}{J}H_1\right)
    \label{eq:lim_I}
\end{split}
\eeq
Now let us take a closer look at $\log_2 p(y) = \log_2(e) \ln p(y)$ which appears
in both $H_0$ and $H_1$.
Let $x_1,x_2 > 0$ with $x_2>x_1$. Then for 
the logarithm of the sum of exponentials it holds that
\beq
-\ln (e^{-x_1} + e^{-x_2}) = x_1 + \ln(1+e^{-(x_2-x_1)}).
\eeq
The error term $\ln(1+e^{-(x_2-x_1)})$ decays exponentially as the difference $x_2-x_1$ grows.
Since $p(y)$ is the sum of two exponentials we can approximate $\ln p(y)$ by:
\beq
\begin{split}
    -\ln p(y) \approx
    \min\left\{\frac{y^2}{2} - \ln(p_0),\frac{1}{2}\left(y-(\eta\hat{P})^\frac{1}{2}\right)^2-\ln(1-p_0)\right\}
\end{split}
\label{eq:max_logsum}
\eeq
This approximation is justified, since the difference of the two exponents
in $p(y)$ is proportional to $\sqrt{J}$, and so it grows large with $J$. 
\footnote{Technically, this approximation does not hold at the point where the two exponents in
    $p(y)$ are equal. However, since the integral of a function does not depend
on the value of the function at points of measure zero, we can redefine $\ln p(y)$ arbitrarily 
at that point.}
First, note, that since $\min\{a,b\} \leq a$ and $\min\{a,b\} \leq b$ holds for all $a,b\in \RR$,
$-\ln p(y) \leq y^2/2-\ln(1-p_0)$ as well as $-\ln p(y+(\eta\hat{P})^\frac{1}{2})\leq y^2/2 + \ln (2^J/K_a)$.
This means that each of the integrands in
$H_0$ and $H_1/J$ resp. is bounded uniformly, for all $J$, by an integrable function.
This allows us to evaluate the integrals by using Lebesgue's theorem on dominated convergence.
For this purpose we need to calculate the pointwise limits of $\ln p(y)$ and
$\ln p(y+(\eta\hat{P})^\frac{1}{2})/J$. The theorem on dominated convergence then states,
that the limit of the integrals is given by the integral of the pointwise limits.\\
The minimum in \eqref{eq:max_logsum} can be expressed as
\beq
-\ln p(y) =
\begin{cases}
    \frac{y^2}{2} &\quad y<\gamma \\
    \frac{1}{2}\left(y-(\eta\hat{P})^\frac{1}{2}\right)^2 +\ln\left(\frac{2^J}{K_a}\right)&\quad y\geq\gamma
\end{cases}
\eeq
where we neglected $\ln(p_0)=\ln(1 - K_a/2^J) \sim K_a/2^J$ and $\gamma$ is given by
\beq
\gamma = \frac{1}{2}\left(\eta\hat{P}\right)^\frac{1}{2} + \ln\left(\frac{2^J}{K_a}\right)\left(\eta\hat{P}\right)^{-\frac{1}{2}}.
\eeq
Given the considered scaling constraints and
$\hat{P} = J\SNR/R_\text{in} = 2J\mathcal{E}_\text{in}$,
$\gamma$ can be rewritten as
\beq
\gamma = \sqrt{\frac{J}{2}}\left(\sqrt{\eta \mathcal{E}_\text{in}} + \frac{1-\frac{1}{\alpha}}{\log e\sqrt{\eta \mathcal{E}_\text{in}}}\right)
\eeq
The term in parenthesis is strictly positive for all $\eta$ so $\lim_{J\to\infty}\gamma = \infty$
and therefore the pointwise limit of $\ln p(y)$ is give by $\lim_{J\to\infty} \ln p(y) = -y^2/2$.
It follows from
Lebesgue's theorem on dominated convergence that
\beq
\lim_{J\to\infty} H_0 = \log_2 e
\label{eq:lim_H0}
\eeq
which is independent of $\eta$, so we can ignore it when evaluating $i^\text{RS}(\eta)$.
For the calculation of $H_1$ note that:
\beq
-\ln p\left(y+(\eta\hat{P})^\frac{1}{2}\right) =
\begin{cases}
    \frac{1}{2}\left(y+(\eta\hat{P})^\frac{1}{2}\right)^2 &\quad y<\gamma'\\
    \frac{y^2}{2} + \ln \left(\frac{2^J}{K_a}\right) &\quad y\geq\gamma'
\end{cases}
\eeq
where we defined
$\gamma' := \gamma - (\eta\hat{P})^\frac{1}{2}$. $\gamma'$ is not non-negative anymore and
therefore the asymptotic behavior of $\gamma'$ depends
on $\eta$ in the following way:
\beq
\lim_{J\to\infty}\gamma' = 
\begin{cases}
    \infty  &\ \text{if } \eta < \bar{\eta}\\
    0       &\ \text{if } \eta = \bar{\eta}\\
    -\infty &\ \text{if } \eta > \bar{\eta}\\
\end{cases}
\eeq
where $\bar{\eta}$ was defined in \eqref{def:eta_bar}.
This gives the following asymptotic behavior:
\beq
-\lim_{J\to\infty}\frac{\ln p(y + (\eta\hat{P})^\frac{1}{2})}{J} = 
\begin{cases}
    \eta \mathcal{E}_\text{in} + \frac{1}{2J} &\ \eta < \bar{\eta}\\
    (1-\alpha^{-1})/\log_2 e &\ \eta \geq \bar{\eta}
\end{cases}
\label{eq:lim_ln}
\eeq
Finally, using \eqref{eq:lim_H0}, \eqref{def:H1}, \eqref{eq:lim_I}, \eqref{eq:lim_ln} and the $\theta$ function defined in
\eqref{def:theta} we get:
\beq
\begin{split}
    &\lim_{J\to\infty}\left(\frac{i^\text{RS}(\eta)}{\log_2 e} - \frac{R_\text{in}2^J}{J}\right) \\
    &= \eta S\mathcal{E}_\text{in}[1-\theta(\eta-\hat{\eta})] \\
    &+ \frac{S}{\log_2 e} \left(1-\frac{1}{\alpha}\right)\theta(\eta-\bar{\eta})
+\frac{1}{2}\left[(\eta - 1) - \ln \eta\right]
\end{split}
\eeq

\newcommand{\mmseor}{\text{mmse}_\text{OR}}
\newcommand{\for}{f^\text{OR}}

This proves the expression \eqref{eq:rs_limit} for the pointwise limit of the potential function.
A similar calculation shows that the derivatives of the potentials, i.e., the functions $\text{mmse}_\text{or}$
also converge pointwise: The $\text{mmse}_\text{or}$ function is given by
\beq
\mmseor(t) = p_0\EE[\for(z)^2] + (1-p_0)\EE[(1-\for(\sqrt{t}+z))^2]
\eeq
where $\for(r)$ was defined in \eqref{eq:for}.
First, we show that the first summand of $\beta\hat{P}\mmseor(\eta\hat{P})$ goes to zero.
Note that $\beta\hat{P} = 2^JR_\text{in}/J\cdot2J\mathcal{E}_\text{in} = \mathcal{O}(2^J/K_a)$
since $R_\text{in} = S/K_a$.
It is apparent that $\for(z) = \mathcal{O}(K_a/2^J)$. Therefore, $\beta\hat{P}\for(z)^2 = \mathcal{O}(K_a/2^J)$,
which goes to zero for all $z$. Furthermore, $\for(z)^2$ can be bounded by $1$ which is integrable w.r.t
the Gaussian density. So, by the dominated convergence theorem, also  $\beta\hat{P}\EE[\for(z)^2] \to 0$.
For the second summands, note that $\beta\hat{P}(1-p_0) = \mathcal{O}(1)$, so it remains to analyze
$\lim \EE[(1-\for(\sqrt{\eta\hat{P}} + z))^2]$. After some algebraic manipulations we find
\beq
\begin{split}
    1-\for\left(\sqrt{\eta\hat{P}} + z\right) = \frac{c}{c' + \exp(\gamma''(z))}
\end{split}
\eeq
where $c,c'>0$ are some constants independent of all variables and
\beq
\gamma''(z) = \eta J\mathcal{E}_\text{in} + J(1-\alpha^{-1}) + z\sqrt{2\eta J\mathcal{E}_\text{in}}.
\eeq

We can see that $\gamma'' \to 0$ for $\eta < \bar{\eta}$, i.e., $\lim (1-\for(\sqrt{\eta\hat{P}} + z))^2 = \text{const.} > 0$,
and $\gamma'' \to \infty$ for $\eta > \bar{\eta}$, i.e., $\lim (1-\for(\sqrt{\eta\hat{P}} + z))^2 = 0$.
Again, the integrand can be bound by $1$, so by the dominated convergence theorem the convergence carries over
to the expected values.

Note that the limiting mmse function is constant on the intervals $(0,\bar{\eta})$ and $(\bar{\eta},\eta)$.
Since both the mutual information and the mmse functions are non-increasing in $\eta$,
a standard result from real analysis states that pointwise convergence implies uniform convergence in
all continuity points of the limit functions. This shows that both the RS-potential and its derivatives converge
uniformly on the intervals $(0,\bar{\eta})$ and $(\bar{\eta},\eta)$. By another standard results from real analysis
this implies that also the sequence of derivatives converges.
This proves the statement of the theorem.

\section*{Acknowledgement} 
We would like to thank the editor and the anonymous reviewers for helpful suggestions
and insightful comments.
PJ has been supported by DFG grant JU 2795/3 and DAAD grant 57417688. AF has
been partially supported by DAAD grant 57512510. AF and GC have been supported
by the Alexander-von-Humbold foundation.

\printbibliography

\end{document}